\newcommand{\longversion}[1]{}
\newcommand{\shortversion}[1]{#1}

\documentclass{article}
\usepackage{ijcai24}

\usepackage{times}
\usepackage{soul}
\usepackage{url}
\usepackage[hidelinks]{hyperref}
\usepackage[utf8]{inputenc}
\usepackage[small]{caption}
\usepackage{graphicx}
\usepackage{booktabs}
\usepackage{algorithm}
\usepackage{algorithmic}
\usepackage{amssymb}
\usepackage[switch]{lineno}

\usepackage{amsmath}
\usepackage{amsthm}
\usepackage{relsize}

\newtheorem{theorem}{Theorem}
\newtheorem{example}[theorem]{Example}
\newtheorem{lemma}[theorem]{Lemma}
\newtheorem{definition}[theorem]{Definition}

\newtheorem{proposition}[theorem]{Proposition}

\shortversion{%
  \urlstyle{same}
}%
\usepackage{mathtools}
\usepackage{microtype}
\DeclarePairedDelimiter\ceil{\lceil}{\rceil}

\DeclareMathOperator{\type}{type}
\newcommand{\intr}{\textit{int}}
\newcommand{\leaf}{\textit{leaf}}
\newcommand{\rem}{\textit{rem}}
\newcommand{\join}{\textit{join}}

{%
  \addtocounter{proposition}{-1}
  \endgroup
}%

\newcommand{\myind}{\phantom{~~}}

\usepackage{etoolbox} %

\newcommand*\linenomathpatch[1]{%
	\cspreto{#1}{\linenomath}%
	\cspreto{#1*}{\linenomath}%
	\csappto{end#1}{\endlinenomath}%
	\csappto{end#1*}{\endlinenomath}%
}

\linenomathpatch{equation}
\linenomathpatch{gather}
\linenomathpatch{multline}
\linenomathpatch{align}
\linenomathpatch{alignat}
\linenomathpatch{flalign}

\title{Epistemic Logic Programs: Non-Ground and Counting Complexity\footnote{This is an author self-archived and extended version of a paper that has been accepted for publication at IJCAI'24~\protect\cite{EiterEtAl24a}.}
}
\author{
Thomas Eiter$^1$
\and
Johannes K. Fichte$^2$\and
Markus Hecher$^3$\and
Stefan Woltran$^1$\\
\affiliations
$^1$Institute of Logic and Computation, TU Wien\\
$^2$Department of Computer and Information Science (IDA), Link\"oping University\\
$^3$Univ. Artois, CNRS, CRIL, France\\
\emails
thomas.eiter@tuwien.ac.at,
johannes.fichte@liu.se,
hecher@cril.fr, 
stefan.woltran@tuwien.ac.at,
}

\usepackage{tikz}
\usepackage{calc}
\usetikzlibrary{calc}

\usepackage{xspace}
\usepackage{amsmath}
\usepackage{color}

\newcommand{\SB}{\ensuremath{\{\,}}
\newcommand{\SE}{\ensuremath{\,\}}}
\newcommand{\SM}{\ensuremath{\;|\;}}
\newcommand{\NAT}{\mathbb{N}}

\DeclareMathOperator{\rootOf}{r}
\DeclareMathOperator{\children}{child-nodes}
\DeclareMathOperator{\dom}{dom}

\DeclareMathOperator{\var}{var}
\DeclareMathOperator{\pname}{pnam}

\DeclareMathOperator{\poly}{poly}

\usepackage{array}
\newcolumntype{H}{>{\setbox0=\hbox\bgroup}c<{\egroup}@{}}

\newcommand{\SAT}{\textsc{SAT}\xspace}

\newcommand{\ELP}{\textsc{ELP}\xspace}

\usepackage{xspace}
\DeclareMathOperator{\eneg}{\textbf{not}}
\newcommand{\gprog}{P}
\newcommand{\prog}{P}
\newcommand{\eprog}{\mathsf{P}}

\newcommand{\vecv}[0]{\mathbf}

\newcommand{\sat}[0]{\text{sat}}

\newcommand{\usat}[0]{\text{usat}}

\DeclareMathOperator{\sign}{sgn}

\newcommand{\QCSPVAL}[1]{\protect\ensuremath{\mathsf{QCSP}_{#1}}}

\DeclareMathOperator{\at}{at}
\DeclareMathOperator{\vars}{vars}
\newcommand{\Card}[1]{\left|#1\right|}

\newcommand{\SIGMA}[1]{\ensuremath{\Sigma^\Ptime_{#1}}\xspace}

\newcommand{\Ptime}{\text{P}\xspace}
\newcommand{\SigmaP}[1]{\ensuremath{\Sigma^\Ptime_{#1}}\xspace}

\newcommand{\QBFSAT}{\textsc{SuccQVal}}
\newcommand{\QBFS}{\textsc{QVal}}

\newcommand{\class}[1]{\textsf{\small#1}}
\newcommand{\Normal}[0]{\class{Normal}}

\newcommand{\Tight}[0]{\class{Tight}}
\newcommand{\NonNeg}[0]{\class{EHorn}\xspace}
\newcommand{\Disj}[0]{\class{Full}}

\newcommand{\kop}{\mathbf{K}\xspace}
\newcommand{\mop}{\mathbf{M}\xspace}

\DeclareMathOperator{\tow}{exp}

\def\hy{\hbox{-}\nobreak\hskip0pt}

\newcommand{\answersets}[1]{AS(#1)}

\newcommand{\calI}{\mathcal{I}}

\newcommand{\eqdef}{\coloneqq}

\newcounter{myenumctr}

\usepackage{cancel}

\usetikzlibrary{matrix}

\tikzset{ 
	table/.style={
		row sep=-\pgflinewidth,
		column sep=-\pgflinewidth,
		nodes={rectangle,%
			align=center,
				anchor=center},
	},
	row 1/.style={nodes={
		}},
}

\DeclareMathAlphabet\mathbfcal{OMS}{cmsy}{b}{n}

\pgfdeclarelayer{bg}    %
\pgfsetlayers{bg,main} 

\let\oldpara\paragraph
\renewcommand{\paragraph}[1]{\oldpara{#1.}}

\DeclareMathOperator{\marg}{arg}
\DeclareMathOperator{\HU}{\mathcal{U}}
\DeclareMathOperator{\HB}{\mathcal{B}}
\DeclareMathOperator{\grd}{grd}
\DeclareMathOperator{\mdef}{def}

\newcommand{\pneg}{\pnot} %

\newcommand{\EXP}{\text{EXP}\xspace}
\newcommand{\NEXP}{\text{NEXP}\xspace}
\newcommand{\EXPH}{\text{EXPH}\xspace}

\newcommand{\NEXPl}[1]{\ensuremath{\NEXP^{\Sigma^\Ptime_{#1}}}\xspace}
\newcommand{\NP}{\ensuremath{\textsc{NP}}\xspace}
\newcommand{\coNP}{\ensuremath{\textsc{coNP}}\xspace}
\newcommand{\sharpP}{\ensuremath{\#\textsc{P}}\xspace}
\newcommand{\sharpd}[1]{\ensuremath{\#\cdot{#1}}\xspace}
\DeclareMathOperator{\co}{co}
\newcommand{\dpc}[1]{\ensuremath{\textsc{D}^{\textsc{P}}_{#1}}\xspace}

\newcommand{\sharpe}{\ensuremath{\#\EXP}}
\newcommand{\sharped}[1]{\ensuremath{\#{\textsc{EXP}}^{#1}}\xspace}

\renewcommand{\phi}{\varphi}

\DeclareMathOperator{\first}{first}
\newcommand{\vx}{\vec{x}}

\newcommand{\vc}{\vec{c}}

\newcommand{\pnot}{\neg} %

\newcommand{\AAA}{\mathcal{A}}
\newcommand{\BBB}{\mathcal{B}}
\newcommand{\LLL}{\mathcal{L}}
\newcommand{\ol}[1]{\overline{#1}}
\DeclareMathOperator{\pdef}{pdef}
\DeclareMathOperator{\lits}{lits}

\DeclareMathOperator{\ar}{ar}
\DeclareMathOperator{\VAR}{VAR}
\DeclareMathOperator{\FO}{FO}
\DeclareMathOperator{\SO}{SO}

\newcommand{\PPP}{\mathcal{P}}
\newcommand{\CCC}{\mathcal{C}}
\newcommand{\VVV}{\mathcal{V}}
\newcommand{\TTT}{\mathcal{T}}
\DeclareMathOperator{\CONS}{CONS}

\newcommand{\extend}[3]{\ensuremath{{#1}\frac{#2}{#3}}}
\DeclareMathOperator{\free}{free}
\newcommand{\BigO}{\mathcal{O}}
\DeclareMathOperator{\suc}{succ}

\longversion{\usepackage[numbers]{natbib}}

\newcommand{\citex}[1]{\citeauthor{#1}~(\citeyear{#1})}
\newcommand{\citey}[1]{\citeauthor{#1} \citeyear{#1}}

\DeclareMathOperator{\enc}{enc}

\DeclareMathOperator{\last}{last}
\DeclareMathOperator{\gat}{gat}

\newcommand{\expf}{\mathrm{exp}}

\allowdisplaybreaks[4]

\DeclareMathOperator{\matr}{\mathsf{matrix}}

\newenvironment{restatecorollary}[1][\unskip]{%
  \begingroup
  \def\thetheorem{\ref{#1}}
}%
{%
  \addtocounter{theorem}{-1}
  \endgroup
}%

\newenvironment{restatelemma}[1][\unskip]{%
  \begingroup
  \def\thetheorem{\ref{#1}}
}%
{%
  \addtocounter{theorem}{-1}
  \endgroup
}%

{%
  \addtocounter{theorem}{-1}
  \endgroup
}%

\newenvironment{restatetheorem}[1][\unskip]{%
  \begingroup
  \def\thetheorem{\ref{#1}}
}%
{%
  \addtocounter{theorem}{-1}
  \endgroup
}%

{%
  \addtocounter{theorem}{-1}
  \endgroup
}%

\begin{document}

\maketitle

\begin{abstract}
  Answer Set Programming (ASP) is a prominent problem-modeling and
  solving framework, whose solutions are called answer sets. 
  Epistemic logic programs (ELP) extend ASP 
  to reason about all or some answer sets.
  Solutions to an ELP can be seen as consequences over multiple
  collections of answer sets, known as world views.
  While the %
  complexity of propositional programs is well
  studied, the non-ground case remains open.
  
  This paper establishes the %
  complexity of non-ground
  ELPs.
  We provide a comprehensive picture for well-known program fragments,
  which turns out to be complete for the class NEXPTIME 
  with access to oracles up to $\Sigma_2^P$.
  In the quantitative setting, %
  we establish complexity results for counting complexity beyond $\sharpe$.
  To mitigate high complexity, we establish 
  results in case of bounded predicate arity, reaching up to the fourth level of the polynomial hierarchy. 
  Finally, we provide ETH-tight runtime results for the parameter treewidth,
  which has applications in quantitative reasoning, where we reason on (marginal) probabilities of epistemic literals.
\end{abstract}

\section{Introduction}
Answer set programming (ASP)
is a widely applied modeling and solving framework for hard
combinatorial problems with roots in non-monotonic reasoning and logic
programming~\cite{BrewkaEiterTruszczynski11} and  solving %
in propositional satisfiability~\cite{FichteBerreHecher23}.
In ASP, knowledge is expressed by means of rules forming a
\emph{(logic) program}. Solutions to those programs are sets of atoms
known as \emph{answer sets}. %
Epistemic logic programs
(ELPs)~\cite{Gelfond91a,KahlWatsonBalai15a,ShenEiter16,Truszczynski11b}
extend ASP by allowing for modal operators \textbf{K} and \textbf{M}, which intuitively mean
``known'' or ``provably true'' and ``possible''
or ``not provably false'', respectively.
These operators can be included into a program and allow for reasoning
over multiple answer sets. Then, solutions to an ELP are known as
world views.

Interestingly, the complexity of decision problems, such as whether a
ground ELP admits a world view, or whether a literal is true in all respectively some
world view, reaches up to the fourth level of the polynomial
hierarchy~\cite{ShenEiter16}.
Despite its hardness in the decision case, also counting world views
is of vivid research interest today %
(see e.g., \cite{BesinHecherWoltran21}), %
as it provides the connection of
quantitative reasoning for ELPs and computing conditional
probabilities by considering the proportion of world views compatible
with a set of literals.
State-of-the-art systems even allow for solving non-ground programs by
either replacing variables with domain constants or structural guided
grounding and then employing existing ASP
solvers~\cite{CabalarEtAl20b,BesinHecherWoltran22}.
Despite the practical implementations and weak known lower
bounds~\cite{DantsinEtAl01,EiterFaberFink07}, the actual %
complexity for non-ground ELPs %
-- and thus the capabilites of today's non-ground ELP systems -- %
remained entirely open.
In particular, it is not known whether epistemic operators in combination with grounding lead to significant 
complexity amplifications or whether we see only a mild increase (reflected by a jump of one level in the PH -- as in the ground case) compared to standard non-ground ASP.
\paragraph{Contributions}
In this paper, we study the precise computational \emph{complexity} of
\emph{qualitative and quantitative} decision and reasoning problems
for \emph{non-ground ELPs}. %
Our contributions are detailed below. In addition,
Table~\ref{tab:results} surveys details and illustrates relations to
existing results.
\begin{enumerate}
  
\item We provide a comprehensive picture of the non-ground ELP
  landscape, including common program fragments. %
  We mitigate %
  complexity by showing how %
  complexity results drop 
  if predicate arities are
  bounded %
  --- a typical assumption for solving.
\item %
We establish detailed complexity results for counting problems, which enables more fine-grained reasoning. To this end, we lift counting complexity notions to the weak-exponential hierarchy.

\item %
We analyze the impact of structural restrictions in form of
  bounded treewidth. If the predicate arities are bounded, we obtain
  precise upper bounds. Surprisingly, the complexity for tight and normal programs match in the non-ground case, which is different to ground programs.
  We complete the upper bounds by conditional \emph{lower bounds}
  assuming  \emph{ETH}\footnote{The exponential time hypothesis (ETH) implies 3-CNF satisfiability cannot be solved in time $2^{o(n)}$ \cite{ImpagliazzoPaturi01}.} rendering significant runtime improvements for treewidth very unlikely.

\end{enumerate}
Interestingly, 
our results are based on two sophisticated techniques.
First, a classical technique employing second-order logic with
dependencies to descriptive complexity for the qualitative setting.
Second, a direct %
approach relying solely on
the validity problem for succinct quantified Boolean formulas, which
enables results for the quantitative setting as well as when
considering bounded treewidth.

\paragraph{Broader Relation to AI}
We see use cases and connections of our results in areas beyond the scope of logic programming. 
In particular, there are complex challenges in, e.g., conformant planning~\cite{Bonet10} or 
in reasoning modes like abduction~\cite{Aliseda17,EiterGottlob95b}, which reach the third and the fourth level of the polynomial hierarchy. 
Such situations can be elegantly modeled via modal operators \textbf{K} and \textbf{M},
even in the non-ground setting. We expect  that the interplay between introspection (i.e., \textbf{K} and \textbf{M} operator capabilities) and non-ground (first-order-like) rules will be of broader interest, as this is essential to formally model rational agents with different belief sets. Here, we provide precise complexity results, consequences of different modeling features, and insights in parameterized complexity. In addition, with the availability of efficient ELP solvers~\cite{BichlerMorakWoltran18,CabalarEtAl20b}, one obtains a first ELP modeling guide. %

\paragraph{Related Work}
\citex{EiterFaberFink07} establish the computational complexity for
qualitative problems of non-ground ASP under bounded predicate
arities.
For ground ELP, \citex{ShenEiter16} show that qualitative problems are
higher up in the Polynomial Hierarchy than for ASP, see Ground case in
Table~\ref{tab:results}. In fact, the central decision problem,
checking whether an ELP has a world view, is $\SigmaP3$\hy complete.
For treewidth and ground ELP, there are solvers that exploit
treewidth~\cite{BichlerMorakWoltran18,HecherMorakWoltran20} and
also solvers for quantitative reasoning, which relate the 
the number of accepting literals to number of compatible 
world views~\cite{BesinHecherWoltran21}.
Very recent works address the grounding bottleneck for solving with
ELP solvers by grounding that exploits
structure~\cite{BesinHecherWoltran23} and complexity of ground ELP
when bounded by treewidth~\citex{FandinnoHecher23}.
Our results reach beyond as we consider the \emph{non-ground
quantitative} and \emph{qualitative} setting.
While the non-ground case might seem somewhat expected, establishing
results on the exponential hierarchy requires different
techniques, especially for~treewidth.
\citex{FichteHecherNadeem22} consider plausibility reasoning in the
ground setting for ASP without epistemic operators.

\newcommand{\mbf}[1]{\mathbf{#1}}
\begin{table*}
  \centering
    \begin{tabular}{l@{\hspace{.15em}}l@{\hspace{.35em}}Hl@{\hspace{.35em}}l@{\hspace{.35em}}l@{\hspace{.35em}}Hl@{\hspace{.15em}}}
      \toprule
                            & \NonNeg         & TBD & \Tight                  & \Normal                   & \Disj                             & Min. WV                        & Result                 \\
      \midrule
      \textbf{Qualitative}  &               &    &                         &                           &                                   &                                &                      \\
      Ground                & $\Ptime$      & \NP      & $\SigmaP2$              & $\SigmaP2$                & $\SigmaP3$                        & $\SigmaP4$-c                   & \cite{Truszczynski11}       \\
      & & & & & & & \cite{ShenEiter16}\\
      \emph{Non-Ground}            & $\mbf{\EXP}$   &  & $\mbf{\NEXP}\mbf{^\NP}$ & $\mbf{\NEXP}\mbf{^{\NP}}$ & $\mbf{\NEXP}\mathbf{^{\SigmaP2}}$ & $\mbf{\NEXP}\mbf{^{\SigmaP3}}$ & \emph{Theorem~\ref{thm:quali}, Lemma~\ref{lem:epihorn}} \\[.05em]	
      \emph{Non-Ground(b)}  & $\mbf{\coNP}$ & & $\mbf{\SigmaP3}$        & $\mbf{\SigmaP3}$          & $\mbf{\SigmaP4}$                  & $\mbf{\SigmaP5}$               & \emph{Theorem~\ref{thm:quali}, Lemma~\ref{lem:epihorn}} \\[.05em]	
      \midrule
      \textbf{Quantitative} &                &   &                         &                           &                                   &                                &                      \\
          \emph{Ground}                & \sharpP        &  & \sharpd{\dpc{}}            & \sharpd{\dpc{}}              & \sharpd{\dpc{2}}               & $\SigmaP4$-c                   & \emph{Lemma~\ref{cor:quantitative:propositional}}       \\
      \emph{Non-Ground}            & $\sharpe$     &   & \sharped{\NP}           & \sharped{\NP}             & \sharped{\SigmaP2}                &                                & \emph{Theorem~\ref{thm:quantitative}} \\[.05em]	
      \emph{Non-Ground(b)}  & \sharpd{\dpc{}}   &  & \sharpd{\dpc{2}}       & \sharpd{\dpc{2}}         & \sharpd{\dpc{3}}                 & $\mathbf{\SigmaP5}$            & \emph{Lemma~\ref{cor:quantitative:barity}} \\[.05em]	
      \midrule
      \textbf{Parameterized}                                                                                                                                                                              \\
				
      Ground [tw] & $\tow(1,{{{{o(tw)}}}})^\dagger$ &$\tow(1,{{{{\Theta(tw)}}}})$ & $\tow(2,{{{{\Theta(tw)}}}})$ & %
                                                                                                                     $\tow(2,{{{{\Theta(tw\cdot \log(tw))}}}})$ %
                                                                                                                                         & $\tow(3,{{{{{\Theta(tw)}}}}})$ & $\tow(4,{{{{{\Theta(tw)}}}}})$ & \cite{FandinnoHecher23}\\
      
      \emph{Non-Gr. [tw](b)} & $\tow(1,\mathbf{{{{d^{o(tw)}}}}})^\dagger$ & $\tow(1,{{{{\Theta(tw)}}}})$  & $\tow(2,\mathbf{{{{d^{\Theta(tw)}}}}})$     & $\tow(2,\mathbf{{{{d^{\Theta(tw)}}}}})$ %
                                                                                                     & $\tow(3,\mathbf{{{{{d^{\Theta(tw)}}}}}})$ & $\tow(4,\mathbf{{{{{{d^{\Theta(tw)}}}}}}})$ & \emph{Theorems~\ref{thm:lower},\ref{thm:upper}} \\ [-.15em]
			\bottomrule
	\end{tabular}
	\caption{
          Complexity results of WV existence (counting/plausibility level) for ELP fragments, where
          each column states the corresponding fragment and each row gives the respective problem.
          ``(b)'' indicates fixed predicate arities.
          Entries indicate completeness results, runtimes are tight under ETH, omitting %
          polynomial factors.
          $d$ refers to the domain size and~$tw$ is the treewidth of the primal~graph. 
          ``$^\dagger$'': The runtime bounds are for the counting case, as decision is easier due to classical complexity results. Here $\expf(0,n)=n$ and $\expf(k,n)= 2^{\expf(k{-1},n)}$, $k\geq 1$.
        }
	\label{tab:results}%
\end{table*}

\section{Preliminaries}
We assume familiarity with basics in %
Boolean satisfiability (\SAT)~\cite{KleineBuningLettman99}. %
By $\tow(\ell,k)$ we refer to~$k$ if~$\ell\leq 0$ and to~$2^{\tow(\ell-1, k)}$ otherwise.

\paragraph{Computational Complexity}
We follow standard notions in 
computational complexity theory~\cite{Papadimitriou94,AroraBarak09}
and use the asymptotic notation~$\BigO(\cdot)$ as usual.
Let $\Sigma$ and $\Sigma'$ be some finite alphabets. We call $I
\in \Sigma^*$ an \emph{instance} and $n$ denotes the size of~$I$.  
A \emph{decision problem} is some subset~$L\subseteq \Sigma^*$. %
Recall that \Ptime{} and \NP are the complexity classes of all
deterministically and non-deterministically polynomial-time solvable
decision problems~\cite{Cook71}, respectively.
We also need the Polynomial Hierarchy
(PH)~\cite{StockmeyerMeyer73,Stockmeyer76,Wrathall76}.
In particular, $\Delta^\Ptime_0 \eqdef \Pi^\Ptime_0 \coloneqq
\Sigma^\Ptime_0 \coloneqq \Ptime$ and $\Delta^\Ptime_{i+1} \coloneqq
P^{\Sigma^p_{i}}$, $\Sigma^\Ptime_{i+1} \coloneqq
\NP^{\Sigma^\Ptime_{i}}$, and $\Pi^\Ptime_{i+1} \coloneqq
\text{co}\NP^{\Sigma^\Ptime_i}$ for $i>0$ where $C^{D}$ is the class~$C$ of
decision problems augmented by an oracle for some complete problem in
class $D$.
The complexity class $\dpc{k}$ is defined as
$\dpc{k}\eqdef \SB L_1 \cap L_2 \SM L_1\in \Sigma^\Ptime_k, L_2 \in
\Pi^\Ptime_k\SE$ and $\dpc{} = \dpc{1}$~\cite{LohreyRosowski23}.
The complexity class~\NEXP is the set of decision problems that can be
solved by a non-deterministic Turing machine using
time~$2^{n^{O(1)}}$,~i.e.,
${\mathsf {NEXPTIME}}=\bigcup _{k\in \mathbb {N} }{\mathsf
  {NTIME}}(2^{n^{k}})$.
The complexity class $\text{NTIME}(f(n))$ is the set of decision
problems that can be solved by a non-deterministic Turing machine
which runs in time $\BigO(f(n))$.
Note that~$\textsc{co-}\NEXP$ is contained in~$\NEXP^\NP$.
The \emph{weak EXP hierarchy (\EXPH)} is defined in terms of
oracle complexity classes:
$\mathsmaller\Sigma_0^{\EXP} \eqdef \EXP$ and
$\mathsmaller\Sigma_{i+1}^{\EXP}\eqdef \NEXP^{\Sigma^p_i}$~\cite{Hemachandra87}.
We follow standard notions in counting complexity~\cite{Valiant79b,DurandHermannKolaitis05,HemaspaandraVollmer95a}. A counting problem is a function~$f: \Sigma^* \rightarrow \NAT_0$.
Then, \sharpP is the class of all functions~$f: \Sigma^* \rightarrow \NAT_0$ such that there is a polynomial-time non-deterministic Turing machine~$M$, where for every instance~$I \in \Sigma^*$, $f(I)$ outputs the number of accepting paths of the Turning machine's computation graph on input~$I$.
We will also make use of classes preceded with the sharp-dot operator `$\#\cdot$' defined using witness functions and respective decision problem in a decision complexity class, see Appendix~\ref{appendix:prelimns:counting}.

\paragraph{Answer Set Programming (ASP)}
Let $(\PPP,\CCC)$ be a first-order vocabulary of 
non-empty finite sets~$\PPP$ of \emph{predicate} and
$\CCC$ of \emph{constant symbols}, and 
let $V$ be a set of \emph{variable symbols}.
{\em Atoms}\/ $a$ have
the form $p(t_1,\ldots,t_n)$, where $p
\in \PPP$, $n \geq 0$ is the arity of $p$, and each $t_i\in \TTT$, where $\TTT = \CCC \cup
\VVV$ is the set of \emph{terms}.
A \emph{logic program (LP)} is a set $\gprog$  of {\em rules}\/ $r$ of the form%
\begin{align*}
    a_1 \vee \ldots \vee a_k \leftarrow
a_{k+1}, \ldots, a_{m}, \pneg a_{m+1}, \ldots, \pneg a_n,
\end{align*}
where all $a_i$ are distinct atoms and
$0\,{\leq}\, k \,{\leq}\, m \,{\leq}\, n$.  We let
$H_r {\eqdef} \{a_1,\ldots, a_k\}$,
$B^+_r {\eqdef}\{a_{k+1},\ldots, a_{m}\}$, and $B^-_r {\eqdef}$
$\{a_{m+1}, \ldots, a_n\}$, and denote the set of \emph{atoms}
occurring in~$r$ and $\gprog$ by
$\at(r) {\eqdef} H_r \cup B^+_r \cup B^-_r$, %
$\at(\gprog){\eqdef} \bigcup_{r\in\gprog} \at(r)$, and
$\pname(P) \eqdef \SB p \SM p(\cdot) \in \at(P) \SE$.
A program has
bounded arity, if every predicate occurring in $P$ has arity at
most~$m$ for some arbitrary but fixed constant~$m$.
By \emph{bounded arity}, refer to the class of programs that are of
bounded arity.
A rule $r$ is a \emph{fact} if $B^+_r\cup B^-_r = \emptyset$; a
\emph{constraint} if $H_r = \emptyset$; \emph{positive} if
$B^-_r = \emptyset$; and \emph{normal} if $\Card{H_r} \leq 1$.  A
program $\gprog$ is positive and normal, respectively, if each
$r\in \gprog$ has the property.  The \emph{(positive) dependency
  graph~$\mathcal{D}_\gprog$} is the digraph with vertices
$\bigcup_{r\in \gprog}H_r \cup B^+_r$, where for each
rule~$r \in \gprog$ two atoms $a{\,\in\,}B^+_r$ and~$b{\,\in\,} H_r$
are joined by edge $(a,b)$. Program~$\gprog$ is \emph{tight}
if~$\mathcal{D}_\gprog$ has no directed cycle~\cite{Fages94}.  We let
$\Disj$, $\Normal$, and $\Tight$ be the classes of all, normal, and
tight programs,~respectively.
Answer sets are defined via ground programs. 
The \emph{arguments} and %
\emph{variables}
of an atom ~$a=p(t_1,$ $\ldots,$ $t_n)$ are the sets $\marg(a) \eqdef
\{t_1,\ldots,t_n\}$ and %
$\vars(a)\eqdef\marg(a)\cap \VVV$. This extends to sets $A$
of atoms by 
$\marg(A)\eqdef \bigcup_{a\in A} \marg(a)$ resp.\ $\vars(A)\eqdef
\bigcup_{a\in A} \vars(a)$ and likewise to rules and programs.
An atom, rule or program $\phi$ is \emph{ground}\/ if $\vars(\phi) =
\emptyset$ and \emph{propositional}\/ if $\marg(\phi) =
\emptyset$. 
The \emph{Herbrand universe} of a program $\gprog$ is the set~$\HU_{\gprog}\eqdef
\marg(\gprog) \cap \CCC$ (if empty, %
$\HU_\gprog \eqdef \{c\}$ for any~$c \in \CCC$), 
and its \emph{Herbrand base}~$\HB_{\gprog}$ consists of all ground
atoms with a predicate from~$P$ and constants from~$\HU_{\gprog}$.

A set~$M \subseteq \HB_P$ of atoms \emph{satisfies} (is a {\em model of}) a
ground rule~$r$ resp.\ ground program $\gprog$ if (i)~$(H_r{\,\cup\,} B^-_r) {\,\cap\,} M {\,\neq\,}
\emptyset$ or (ii)~$B^+_r {\,\setminus\,} M {\,\neq\,}
\emptyset$ resp.\ $M$ satisfies each $r\in \gprog$.
Furthermore, $M$ is an \emph{answer set} of $\gprog$ if $M$ is a
$\subseteq$-minimal model of~$\gprog^M \eqdef \bigcup_{r\in\prog}\SB H_r \leftarrow B^+_r \SM B^-_r
\cap M = \emptyset\SE$,~i.e., the \emph{GL-reduct} of $\gprog$~\cite{GelfondLifschitz91} w.r.t.~$M$; $\answersets{\gprog}$ denotes the set of all answer sets of $\gprog$.
The answer sets of a general program $\gprog$ are those of
its \emph{grounding} $\grd(\gprog) \eqdef \bigcup_{r
\in \gprog} \grd(r)$, where $\grd(r)$
is the set of all rules obtained by replacing each~$v \in
\vars(r)$ with some element from~$\HU_\gprog$.
We assume {\em safety},
i.e.,
each rule $r\,{\in}\, \gprog$ satisfies $\vars(H_r \cup B^-_r) \subseteq \vars(B^+_r)$.
We can ensure it by
a unary domain predicate $\mathit{dom}$ 
with facts $\mathit{dom}(c)$, $c\in \HU_{\gprog}$ and
adding $\mathit{dom}(x)$ in the body of $r$
for each $x\in\vars(r)$.
To select rules in $\grd(\gprog)$ with the same
head~$D$, we define $\mdef(D,\gprog) \eqdef \{r \in \grd(\gprog) \SM
H_r= D\}$ and to select non-ground rules
in~$\gprog$ that define atoms with predicate
$p$, we let $\pdef(p,\gprog) \eqdef \SB r \in \gprog \SM
p(t_1,\ldots,t_n) \in H_r \SE$.
Deciding whether a program $\gprog$ has an answer set (called
\emph{consistency}) is
\SIGMA{2}{}-complete for ground programs~\cite{EiterGottlob95} and $\NEXP^\NP$ for non-ground programs~\cite{EiterGottlobMannila94}.

\noindent
\textbf{Epistemic Logic Programs  (\ELP{s}).}
\emph{Epistemic logic programs} extend LPs with epistemic literals in rule bodies.
A \emph{literal} is either an atom $a$ (\emph{positive literal}) or
its negation $\neg a$ (\emph{negative literal}). 
A set~$L$ of literals is \emph{consistent} if
for every~$\ell \in L$, $\neg \ell \not \in L$ assuming that
$\neg \neg \ell = \ell$. 
For a set~$A$ of atoms, we define
$\neg A \eqdef \SB \neg a \SM a \in A\SE$  and
$\lits(P) \eqdef \at(P) \cup \neg \at(P)$.
An \emph{epistemic literal} an expression~$\eneg \ell$ where $\ell$ is
a literal.
Following common convention, we use $\kop \ell$ as shorthand for $\neg \eneg \ell$ and $\mop \ell$
for $\eneg \neg \ell$.
An \emph{epistemic atom} is an atom that is used in an epistemic
literal.
For a set~$S$ consisting of atoms, literals, and/or epistemic
literals, we denote by $\at(S)$ and %
$\lits(S)$, %
the set of atoms and literals, %
respectively, that occur in $S$.
These notations naturally extend to rules and programs.
Definitions for logic programs such as classes of programs
naturally extend to ELP.
The \emph{dependency graph~$D_P$} of an ELP~$P$ is as for ASP, but for every rule~$r\in P$ and~$b\in H_r$, we also add an edge~$(a,b)$ if~$r$ contains a body literal $\eneg \neg a$ or $\neg \eneg a$. %
Properties are similar to ASP.
In addition, we define \NonNeg\ with no negations (neither $\neg$ nor $\eneg$) and no disjunctions, however, $\kop a$ and $\mop a$ are allowed.
There are different
semantics for
ELPs~, e.g.\ \cite{Gelfond91a,Truszczynski11,KahlWatsonBalai15a,ijcai:CerroHS15,ShenEiter16}; see 
\cite{DBLP:journals/tplp/FandinnoFG22} for an overview.
We consider \cite{ShenEiter16}, which provides a reduct-based framework and offers highest problem solving capacity.

In what follows, 
let $\eprog$ be a ground ELP. %
A \emph{world view interpretation} (\emph{WVI}) for
$\eprog$ is a consistent set $I \subseteq \lits(\eprog)$.
Intuitively, every~$\ell\in I$ is considered ``known'' and
every $a\in \at(\eprog)$ with~$\{a,\neg a\}\cap I=\emptyset$ is treated as
``possible''.
The \emph{epistemic reduct}~\cite{ShenEiter16,iclp:Morak19} of
program~$\eprog$ under WVI~$I$ is $\eprog^I \eqdef \{r^I \SM r \in
\eprog\}$,
where $r^I$ results by replacing in $r$ 
each epistemic literal $\eneg \ell$ with $\neg \ell$ if
$\ell \in I$ and with $\top$ otherwise; %
double negation cancels. %
This amounts to FLP-semantics for nested negation; we omit HT-semantics, for which similar complexity results can be obtained.
Note that $\eprog^{I}$ 
has no epistemic negations.
A WVI $I$ over~$\lits(\eprog)$ is \emph{compatible} with a set $\calI$  of
WVIs if
(i)~$\calI \neq \emptyset$ and for each atom $a$,
(ii)~$a \in I$ implies $a \in \bigcap_{J\in \calI}J$;
(iii)~$\neg a \in I$ implies $\{ J\in \calI \mid a \in J\}=\emptyset$; and
(iv)~$a \in \at(\eprog) \setminus \at(I)$ implies
that $a \in J$ and $a \not\in J'$.
for some $J, J' \in \calI$.
$I$ is a \emph{candidate world view (WV)}  %
of $\eprog$ if $I$ is compatible with the set~$\answersets{\eprog^I}$.
\emph{WV existence} is $\Sigma_{3}^{\Ptime}$-complete~\cite{Truszczynski11,ShenEiter16}.
\newcommand{\WVc}[0]{\ensuremath{\#\text{WV}}\xspace}
The counting problem~\WVc asks to output the number of WVs. 
Semantics of non-ground ELPs is defined by grounding, as for LPs.

\newcommand{\pr}[1]{\text{#1}}
\begin{example}[cf.~\citey{Gelfond91a}]\label{ex:elp}
   Take the well-known scholarship eligibility problem encoding, which is as follows:\\
 \myind   $P_1 \,{=}\,  \SB \mbox{~} \pr{lowGPA}(mark); \quad \pr{highGPA}(mia);$\\
 \myind    \mbox{~}$\qquad\quad\pr{lowGPA}(maya) \vee \pr{highGPA}(maya);$\\
 \myind   \mbox{~}$\qquad\quad \pr{inelig}(X) \leftarrow \pr{lowGPA}(X);\;$\\
 \myind    \mbox{~}$\qquad\quad\pr{elig}(X) \leftarrow \pr{highGPA}(X);$\\ 
 \myind   \mbox{~}$\qquad\quad \bot \leftarrow \pr{elig}(X ), \pr{inelig}(X);\;$\\ \myind    \mbox{~}$\qquad\quad\pr{interview}(X) \leftarrow \eneg \pr{elig}(X), \eneg \pr{inelig}(X)\;\}$.
\\ %
    Then, the set of WVs of the program is\\%
\myind     $\{\;\{\neg \pr{interview}(mark), \pr{lowGPA}(mark), \pr{inelig}(mark),$\\
\myind\myind\myind~~$\neg \pr{elig}(mark), \pr{interview}(maya), \neg \pr{interview}(mia),$\\
\myind\myind\myind~~$\pr{highGPA}(mia),
     \pr{elig}(mia),  \neg \pr{inelig}(mia)\}\;\}$.
\end{example}

\paragraph{First- and Second-Order Logic}
We assume familiarity with logic and follow standard
definitions~\cite{GraedelEtAl07} %
(see also the supplemental material). %
Throughout 
we assume that $\sigma$ is a signature, which we
omit if it is clear from context. The class~$\Sigma^1_k[\sigma]$ consists of all {\em prenex second-order  formulas} $\Phi\in \SO[\sigma]$, i.e., 
$\Phi = Q_1 R_1 Q_2 R_2 %
\cdots Q_k R_k. \phi$ where
$Q_i \in \{\forall, \exists\}$ and $Q_i \neq Q_{i+1}$ for
$1 \leq i <k$, 
the $R_i$ are disjoint non-empty sets of
SO-variables, and $\phi\in \FO[\sigma]$; $\Phi$ is \emph{existential}
if $Q_1 = \exists$.
We say that $\Phi$ is in \emph{CDNF} if $\free(\phi)=\emptyset$ and
(i) $k$ is even and $\phi = \exists \vx \psi$
with $\psi$ in DNF, or (ii)  $k$ is odd and $\phi = \forall \vx\psi$
with $\psi$ in CNF.

\section{Complexity of Non-ground ELP Reasoning}
In this section, we establish results on the classical complexity of
reasoning with non-ground ELPs.
Our first insight is on qualitative reasoning.
Therefore, we need Proposition~\ref{prop:socompl}, which states the
relationship between second-order logic and the exponential hierarchy
for combined and data complexity.

\begin{proposition}[\citey{GottlobLeoneVeith99}]\label{prop:socompl}
  Given a sentence~$\Phi\in \Sigma^1_k$ and a
  finite structure~$\AAA$, 
  deciding whether $\AAA \models \Phi$  is
  (i)~$\NEXPl{k-1}$-complete (combined complexity) and
  (ii)~$\SIGMA{k}$-complete if $\Phi$ is fixed (data complexity).
\end{proposition}

Next, in Lemma~\ref{lem:transform}, we show a connection between the
existing result on second-order logic and the exponential hierarchy in
the general case and in case predicates have bounded arity.

\begin{lemma}[$\star$\footnote{
    We prove statements marked by ($\star$) in the appendix.
  }]\label{lem:transform}
  Given a sentence~%
  $\Phi\in \Sigma^1_k[\sigma]$ in CDNF and a
  finite structure~$\AAA$, %
  deciding whether $\AAA \models \Phi$  is
  (i)~$\NEXPl{k-1}$-com\-plete and
  (ii)~$\SIGMA{k+1}$-complete if every 
  predicate $R_i$ in $\Phi$ has arity at most~$m$ for some arbitrary but
  fixed integer~$m\geq 1$.
\end{lemma}

\subsection{Qualitative Reasoning}
With the help of the results above, we are ready to establish the following central insight into the complexity of non-ground ELPs. 
While it turns out that world view existence on a limited fragment is already complete for a class beyond \NEXP, luckily, for bounded predicate arity we obtain completeness results for the fourth level of the polynomial hierarchy.

\begin{theorem}\label{thm:quali}
  Let~$\prog$ be an ELP and 
 (a)\,  $i=2$ if $P \in \Disj$, and (b)\, $i=1$ if $P \in \Normal \cup \Tight$.
  Then, deciding whether~$\prog$ admits a world view is 
  \NEXPl{i}-c for non-ground $P$ and 
  $\SIGMA{i+2}$-c for non-ground of  \emph{bounded arity}.
\end{theorem}
\begin{proof}[Proof (Sketch)]
  \emph{Membership:} For the non-ground cases, %
  the result follows 
  immediately  from the $\Sigma^p_{i+1}$-completeness %
  in the  ground (propositional) case~\cite{ShenEiter16}, as grounding an ELP $\eprog$ leads to an exponentially 
  larger program $\grd(\eprog)$ and $\Sigma^p_{j}$ becomes $\NEXPl{j-1}$~\cite{GottlobLeoneVeith99}.
  For the bounded arity cases. %
  If predicate arities are bounded by a constant, a guess for a WVI~$I$  
  of an epistemic program $\eprog$ has polynomial size. We can emulate the epistemic reduct~$\eprog^I$ 
  by replacing in $\eprog$ each epistemic literal $\eneg \ell$ where 
  $\ell =  L(\vec{t})$, $L \in \{p, \pnot p\}$ by an atom $q_{\eneg L}(\vec{t})$,  where 
  $_{\eneg L}$ is a fresh predicate of arity $|\vec{t}|$, and add the following fact or rule, for each tuple $\vec{c}$ of constants (having arity $|\vec{t}|$: 
(1) $q_{\eneg L}(\vec{c})$, if $L(\vec{c})\notin I$, and 
(2)   $q_{\eneg L}(\vec{c}) \leftarrow \pneg{}L(\vec{c})$ otherwise
  (double negation cancels).
  Then the answer sets of the resulting program $\eprog_{I,\eneg}$ correspond to the answer sets  of $\eprog^I$, as (1) and (2), respectively,
  can be unfolded with rules in the grounding of $\eprog$ that contain $\eneg L(\vec{c})$. 
   In particular,  $I$ is compatible with $\answersets{\eprog^I}$ iff $I$ is compatible with $\answersets{\eprog_{I,\eneg}}$. %
  As $\eprog_{I,\eneg}$ has bounded predicate arity, brave and cautious reasoning
  from $\eprog_{I,p_{\eneg}}$ is in $\Sigma^p_3$ and $\Pi^p_3$, respectively,~\cite{EiterFaberFink07}. %
  Consequently, we can check in  polynomial time with an $\Sigma^p_3$ oracle whether $I$ fulfills conditions (i)--(ii) of compatibility with $\answersets{\eprog^I}$,~i.e., 
  whether $I$ is a WVI of~$\eprog$. This shows membership in $\Sigma^p_4$, i.e., $\eprog \in \Disj$. %
  If $P \in \Tight \cup \Normal$, %
  brave and cautious reasoning from $\eprog_{I,p_{\eneg}}$ is 
  in~$\Sigma^p_2$ and $\Pi^p_2$, respectively,~\cite{EiterFaberFink07}, as program~$\eprog_{I,p_{\eneg}}$ is 
  normal/tight, if program~$\eprog$ is so. 
  This shows  membership in $\Sigma^p_3$.

  \noindent\emph{Hardness:}
  We construct from a given sentence~$\Phi\in \Sigma^1_k$ and
  finite structure $\AAA$ an ELP~$\prog$,
  thereby, we reduce deciding whether $\AAA \models \Phi$ (model
  checking) to deciding whether~$\prog$ admits a world view
  (world-view-existence).
  In our reduction, we lift the existing ELP encoding that solves QBF
  validity to SO~\cite{ShenEiter16}.
  In detail,
  Case~$P \in \Disj$:
  let  $\AAA=(A,\sigma^\AAA)$ and
  $\Phi = \exists R_1\forall R_2\exists R_3.\phi \in \Sigma^1_k$ where
  $\phi = \forall \vx\psi$ and
  $\psi = \bigwedge_{j=1}^m\bigvee_{h=1}^{\ell_j}L_{j,h}$ is in
  CNF,~i.e., as in Lemma~\ref{lem:transform}.
  We take $u$ and $v$ as propositional atoms and 
  for each relation symbol~$R \in R_1 \cup R_2 \cup R_2$, we introduce
  predicates $R$ and $\ol{R}$. Then, we construct programs $\prog_1,\ldots,\prog_5$ as follows. Let $e(\vx) = (e(x_1), e(x_2), \ldots, e(x_m))$ and
  $e(v) \eqdef v$, if $v$ is an FO-variable and $e(v) \eqdef c^\AAA$
  if $v$ is a constant symbol.
  Intuitively, $\prog_1$ selects with epistemic negation a candidate
  world view corresponding to 
  a guess  for each relation symbol $R_{1,i}$
  in~$R_1$ using an auxiliary relation symbol $\ol{R_{1,i}}$ for its complement.%
  \\[0.25em]
  \myind\noindent $\prog_1 = \{ R_{1,i}(e(\vx_{1,i})) \leftarrow  \eneg \ol{R_{1,i}}(e(\vx_{1,i}));$\\
   \myind${}\qquad\quad\ol{R_{1,i}}(e(\vx_{1,i})) \leftarrow \eneg
  R_{1,i}(e(\vx_{1,i})) \SM R_{1,i} \in R_1 \}$.\\[0.25em]
  Program~$\prog_2$ generates then answer sets for each possible relation $R_{2,i}$ in~$R_2$.\\[0.25em]
   \myind$\prog_2 =  \{ \ol{R_{2,i}}(e(\vx_{2,i})) \leftarrow  \pneg R_{2,i}(e(\vx_{2,i}));$\\
   \myind${}\qquad\quad \ol{R_{2,i}} (e(\vx_{2,i})) \leftarrow \pneg
  \ol{R_{2,i}} \mid R_{2,i} \in R_2\}$ 
  \\[0.25em]
  Programs $\prog_3$ guesses for each such valuation of $R_2$ a valuation of $R_3$ such that $\forall \vx\psi$ is true.
  \\[0.25em]
   \myind\noindent$\prog_3 = \{ R_{3,i}(e(\vx_{3,i})) \vee
  \ol{R_{3,i}}(e(x_{3,i})) \SM R_{3,i} \in R_3\}$\\[0.25em]
  The program $\prog_4$ checks then using the saturation technique that $\forall\vx\psi$ is not violated, i.e., $\exists\vx \neg \psi$ is false; any violation makes $u$ true and saturates the guess.
  The last rule in $\prog_4$ eliminates the candidate world view if
  $\neg u$ cannot be derived.\\[0.25em]
  \noindent$\prog_4 = %
  \{u \leftarrow s(L_{j,1}), \ldots, s(L_{j,\ell_j}) \SM 1 \leq j\leq m\} \,\cup$\\
  \mbox{${}\qquad\;\,\{R_{3,i}(e(\vx_{3,i})) \leftarrow u;\; %
  \ol{R_{3,i}}(e(\vx_{3,i})) \leftarrow u \SM R_{3,i} \,{\in}\, R_3\} \cup$} \\
  ${}\qquad\;\,\{v \leftarrow \eneg v,\eneg \pneg u \}.$\\[0.25em]
  where $s(R(\cdot)) \eqdef \pneg R(\cdot) $ if $R \notin \sigma$,
  $s(R(\cdot)) \eqdef R(\cdot)$ otherwise and
  $s(\neg R(\cdot)) \eqdef R(\cdot)$.\\[0.25em]
  Program~$\prog_5$ represents atoms of the input structure~$\AAA$ as
  facts.\\[0.25em]
   \myind$\prog_5 = \{R_i(\vc) \SM \vc \in R_i^\AAA,\vc\in U^{|\vc|},
  1 \leq i\leq k\}.$\\[0.25em]
  Notably, we treat equality as the other relations.
  Finally, we build the program $\prog = \bigcup_{i=1}^5 \prog_i$. Then $P$
  is constructible in polynomial time and it has a world view iff  $\AAA \models \Phi$.\\[0.25em]
  Case~$P \in \Tight \cup \Normal$:
  We encode evaluating an SO sentence $\Phi = \exists P_1\forall P_2\exists \vx
  \psi$ over $\AAA$. We assume $\psi =
  \bigvee_{j=1}^m\bigwedge_{h=1}^{\ell_j}L_{j,h}$ is a DNF, drop
  $\prog_3$, and replace $\prog_4$ with the following rules:\\[0.25em]
   \myind$\prog_4'
  = \{u\leftarrow L_{j,1},\ldots,L_{j,\ell_j} \SM 1 \leq
  j\leq m\} \,\cup$ \\
   \myind${}\qquad\;\,\{v \leftarrow \eneg v, \eneg u \}$.
  \\[0.25em]
  For each valuation of $P_1$ and
  $P_2$, we have then a unique answer set that contains
  $u$ iff $(\AAA,R^\AAA_1,R^\AAA_2) \models \exists
  \vx\psi$.
  Then $\prog' = \prog - (\prog_3 \cup \prog_4)\cup
  \prog_4'$ has a world view iff $\AAA \models
  \prog$. As $\prog_4'$ is tight, the reduction applies for tight programs as well.
 \end{proof}

\begin{lemma}[$\star$]\label{lem:epihorn}
  Let~$\prog \in \NonNeg$. %
  Then, deciding whether~$\prog$ admits a world view is
  in \Ptime{} if $P$ is ground,
    \coNP-complete if $P$ is non-ground and has bounded arity, and
    \EXP-complete if $P$ is non-ground.
\end{lemma}

\subsection{Counting Complexity Beyond \NEXP}
Before we can turn our attention to the quantitative setting, 
we need to define counting classes for classifying counting problems, whose corresponding decision problems are in $\NEXP^\mathcal{C}$ for a decision class~$\mathcal{C}$.
Following, we provide canonical problems, followed by completeness results for ELPs.

\paragraph{Generalizing Counting Classes} 
For counting solutions of problems in \NEXP, the corresponding counting complexity class 
\#EXP~\cite{PapadimitriouYannakakis86} has been defined. However, classes based on oracle machine models, as in the $\#\cdot$ complexity classes~\cite{HemaspaandraVollmer95a} have been left out for exponential time. 
 Below, we generalize counting complexity to the realm of exponential time. 
 This allows us to describe in analogy to decision problems the complexity of counting problems beyond \NEXP. %

\begin{definition}[Exp-Oracle Classes]
    Let~$\mathcal{C}$ be a decision complexity class.
    Then, 
    $\sharped{\mathcal{C}}$ is the class of counting problems, whose solution is obtained by counting the number of accepting paths %
    of a non-deterministic Turing machine in exponential time with access to a $\mathcal{C}$ oracle. %
\end{definition}

Observe that by construction $\sharpe = \sharped{\Ptime}$.
To demonstrate these classes, we define a family of counting problems serving as canonical representatives.

\paragraph{Succinct Quantified Boolean Formulas}
To define succinct formula representation, we vastly follow existing ideas~\cite{Williams08}.
For a circuit~$C$ with a set~$I$ of~$n$ inputs, where $C$ has size~$\poly(n)$, we let~$T(C)$ be the \emph{truth table} of the Boolean function represented by~$C$.
Formally, $T(C)$ is the $2^n$-bit string such that~$T(C)[i]=C(B_i)$, where $B_i$ is the $i$-th of all $n$-bit strings in lexicographic order;
intuitively, $T(C)[i]$ is bit $i$ of a string that encodes a problem instance.

For a 3CNF $\varphi$, we define such a circuit~$C_\varphi$ (``\emph{clause circuit}'') over \emph{sign-bits}~$s_j$ for $j\in[1,3]$ and \emph{variable-bits}~$b_j^k$ for $k \in[1,v_\phi]$, where $v_\phi=\ceil{\log(|\var(\varphi)|)}$. 
More precisely, $C_\varphi$ has $3(v_\phi+1)$  input bits~$s_1, \vec{b_1},
s_2, 
\vec{b_2},
s_3,
\vec{b_3}$, 
where $s_j$ tells whether the~$j$-th literal $\ell_j$ in a 3CNF clause, whose variable is encoded by the bits $\vec{b_j} = b_j^1, \ldots, b_j^{v_\phi}$, is positive ($s_j=1$) or negative ($s_j=0$).
That is, $\ell_1 \vee \ell_2 \vee \ell_3\in \varphi$ if and only if %
$C_\varphi(\sign(\ell_1), \vec{b_1}, \sign(\ell_2), \vec{b_2},\sign(\ell_3), \vec{b_3})=1$. 
For~$\varphi$ in 3DNF, a circuit~$C_\varphi$ (``\emph{term circuit}'')  is defined analogously.

\begin{example}
    Let~$(a \vee b \vee \neg c) \wedge (\neg b \vee a \vee d) \wedge (\neg b \vee c \vee \neg d)$ be a Boolean formula in 3CNF.
    Using $2=\ceil{(\log(4))}$ bits, we can succinctly represent this formula as a circuit.
 \end{example}

For QBFs, we must also succinctly represent quantifiers. While we could merge this into the clause or term circuit,
for the sake of readability, we use a second circuit.
For a QBF~$Q=\exists V_1. \forall V_2. \ldots\allowbreak Q_\ell V_\ell. \varphi$ with alternating  quantifier blocks $Q_i\,{\in}\, \{\exists, \forall\}$, we
define a \emph{quantifier circuit} $C_Q$ 
with $\ceil{\log(l)}\,{+}\, v_\phi$ many input bits $\vec{q},\vec{b}$, where 
$\vec{q}\,{=}\, q^1,\ldots,q^{\ceil{\log(l)}}$ and $\vec{b}\,{=}\, b^1,\ldots,b^{v_\phi}$. 
Intuitively, 
$v\,{\in}\,\var(\varphi)$ is in~$V_\iota$ iff  $C_Q(bin(\iota),$ $bin(v))=1$,
where~$bin(\cdot)$ is the binary representation.
$Q$ is \emph{closed} if every~$v\in\var(\varphi)$ is in~$V_\iota$ for some~$\iota$; otherwise~$Q$ is \emph{open} and its (set of) \emph{free variables}~$\var(\varphi)\setminus(\bigcup_{1\leq \iota \leq \ell} V_\iota)$.

\begin{definition}[Succinct QBF]
    A \emph{succinct QBF} $Q$ with $\ell$ alternating quantifier blocks (alternation depth) is given by a \emph{quantifier circuit $C_Q$ and a clause circuit $C_\phi$}. 
    Problem $\QBFSAT_\ell$ is deciding whether a closed 
    succinct QBF $Q$ evaluates to true; $\#\QBFSAT_\ell$ asks to count %
 assignments $\theta$ over the free variables of~$Q$ such that $Q\theta$ evaluates to true.
\end{definition}

The following complexity result is known.

\begin{proposition}[Complexity of $\QBFSAT_\ell$~\cite{GottlobLeoneVeith99,Stewart91}]\label{prop:succ}
For succinct QBFs $Q$ of alternation depth $\ell\geq 1$,
$\QBFSAT_\ell$
is~$\NEXP^{\SigmaP{\ell{-}1}}$-complete.
\end{proposition}

\noindent
This immediately yields corresponding counting complexity. %

\begin{proposition}[Complexity of $\#\QBFSAT_\ell$]\label{prop:countsucc}
For succinct QBFs $Q$ of alternation depth $\ell\geq 0$,
counting the number of assignments over its free variables under which $Q$ evaluates to true is $\sharped{\SigmaP{\ell}}$-complete. %
\end{proposition}

We will utilize this result by defining a parsimonious reduction to our counting problems of interest. A \emph{parsimonious reduction} is a polynomial-time reduction from one problem to
another that preserves the number of solutions, i.e., it induces a bijection
between respective sets of solutions of two problems.

\subsection{Quantitative Reasoning}
Next, we discuss how  quantitative aspects enable more fine-grained reasoning. 
Indeed, deciding whether a world view exists concerns only a single world view.
Instead, if we aim for stable results 
towards consensus among different world views, 
one would prefer computing \emph{levels of plausibility} for certain observations (assumptions). This is achieved by quantitative reasoning, where we quantify the number of world views satisfying a given query~$Q$. %
Thereby we compute the level of plausibility for $Q$. We need the following notation.

\begin{definition} 
    An \emph{(epistemic) query} is a set of expressions of the form $\mop \ell$ or $\kop \ell$, where $\ell$ is a literal. %
\end{definition}

\noindent Intuitively, we can then quantify the plausibility of
queries. To this end, we define the union of an ELP~$\eprog$ and a
query~$Q$, which is a set of expressions as defined above,
by~$\eprog\sqcup Q\eqdef \eprog\cup \{v\leftarrow \eneg v, \neg q \mid
q\in Q\}$ for a fresh atom $v$.

\begin{definition}[Plausibility Level]
    Let~$Q$ be an epistemic query and $\eprog$ be an ELP. The \emph{plausibility level of~$Q$} is defined as $L(\eprog, Q)\eqdef \WVc(\eprog \sqcup Q)$. %
\end{definition}

\noindent We define probabilities via two counting operations and therefore study the complexity of computing plausibility levels.

\begin{definition}[Probability]
    Let~$Q$ be an epistemic query and $\eprog$ be an ELP. The \emph{probability of~$Q$} is defined as $\frac{L(\eprog, Q)}{\max(1,L(\eprog, \emptyset))}$.
\end{definition}

\noindent Observe that the empty query has probability $1.0$ and inconsistent queries or ELPs both have probability $0.0$ (implausible). 

\paragraph{Complexity of Computing Plausibility Levels}
For establishing the complexity of counting, we reduce from $\#\QBFSAT_\ell$ and use Proposition~\ref{prop:countsucc}.
Indeed, computing plausibility levels is already hard for empty queries.

\begin{theorem}[$\star$]\label{thm:quantitative}
  Let~$\eprog$ be an ELP and 
 (a)\, $i=2$ if $\eprog \,{\in}\, \Disj$, 
 (b) $i=1$ if $\eprog \,{\in}\, \Normal \cup \Tight$, 
  and
(c)\, $i=0$ if $\eprog \,{\in}\, \NonNeg$.
  Then, computing plausibility level~$L(\eprog,\emptyset)$ is
  $\sharped{\SIGMA{i}}$-complete. 
\end{theorem}
\begin{proof}
\noindent\emph{Membership:}
This follows immediately from the ground case (propositional) 
where we have $\Sigma^p_{i+1}$-comple\-te\-ness~\cite{ShenEiter16}.
With the same argument as in the proof of Theorem~\ref{thm:quantitative}, meaning, grounding an ELP~$\eprog$ leads to an exponentially larger program~$\grd(\eprog)$, and $\Sigma^p_{j}$ becomes $\NEXPl{j-1}$~cf.~\cite{GottlobLeoneVeith99}, we conclude the result.

\smallskip
\noindent\emph{Hardness for $\eprog \in \NonNeg$:}\\
We reduce from a \emph{restricted fragment} of~$\#\QBFSAT_0$, taking a positive Boolean formula $\varphi$ defined by a clause circuit~$C$
over~$3\cdot (1+n)$ many input gates, and constructing an ELP~$\eprog$.

First, the evaluation of $C$ is inductively constructed, starting from the input gates of~$C$ to the output gate of~$C$.
Thereby, for every gate~$g$ we construct a rule defining a predicate of arity~$3\cdot(1+n)$,
depending on the result of the predicates for the input gates of~$g$.
By~$\vec v_i$ we refer to a sequence of~$n$ many variables~$v_i^1, \ldots, v_i^n$.
Also, we define the facts $b(0)$ and $b(1)$.
For an input gate~$g_j$ of~$C$ with~$1\leq j \leq 3\cdot (1+n)$, we construct the fact~$g_j(v^1, \ldots, v^{3\cdot (1+n)}).$ if and only if $v_j = 1$.
Without loss of generality, we assume that negation only appears at the input gates (negation normal form).

For a conjunction gate $g_\wedge$ with inputs~$g_1, \ldots, g_o$,
we define\\
\myind$g_\wedge(s_1, \vec v_1, s_2, \vec v_2, s_3, \vec v_3) \leftarrow
g_1(s_1, \vec v_1, s_2, \vec v_2, s_3, \vec v_3),\ldots, $\\
\myind\mbox{~}\hspace{11em}$g_o(s_1, \vec v_1, s_2, \vec v_2, s_3, \vec
v_3)$;
\\[0.25em]
for disjunction gate $g_\vee$ with inputs~$g_1, \ldots, g_o$, we
define for every~$1\leq k\leq o$:\\
$g_\vee(s_1, \vec v_1, s_2, \vec v_2, s_3, \vec v_3) \leftarrow
g_k(s_1, \vec v_1, s_2, \vec v_2, s_3, \vec v_3).$\\[0.25em]
We refer to the predicate of the final output gate of the construction
by~$g_C$.
Additionally, we construct the following rules below.
First, we guess an assignment over the variables, where we decide whether a variable will be set to false:\\
\myind$\dot V(v^1, \ldots, v^{n}) \leftarrow \mop \dot  V(v^1, \ldots, v^{n}), b(v^1), \ldots, b(v^n).$%

\noindent Then, we check whether there is an unsatisfied clause.\\
\noindent $%
\myind\leftarrow \dot V(\vec v_1), g_C(1, \vec v_1, s_2, \vec v_2, s_3, \vec v_3), %
\dot V(\vec v_2), $

\noindent $\qquad g_C(s_1, \vec v_1, 1, \vec v_2, s_3, \vec v_3), %
\dot V(\vec v_3),g_C(s_1, \vec v_1, s_2, \vec v_2, 1, \vec v_3).$\\

It
is easy to see that there is a bijection between satisfying assignments of $\varphi$ and world views of $\eprog$.

\medskip
\noindent\emph{Hardness for $\eprog \in \Tight \cup \Normal$}:\\
We reduce from~$\#\QBFSAT_1$, taking a QBF $Q=\forall U. \varphi$ over free variables~$V$, with~3DNF~$\varphi$ given by a term circuit~$C$
over~$3\cdot (1+n)$ many input gates and a quantifier circuit~$D$ over~$n$ input gates. From this, we construct ELP~$\prog$.

First, $C$ is inductively constructed as above. %
We refer to the predicate of the output gate of the construction by~$g_C$.
Then, similar to above, we define for every gate of the circuit~$D$ a predicate of arity~$n+1$ from the input gates
to the output gate of~$D$.
For a negation gate $g_\neg$ with input~$g$, we define\\
\myind$g_\neg(\vec v_1) \leftarrow \pneg g(\iota, \vec v_1),
b(\iota),\allowbreak b(v_1^1), \ldots, b(v_1^n)$;\\[0.25em]
for a conjunction gate $g_\wedge$ with inputs~$g_1, \ldots, g_o$,
we define\\
\myind$g_\wedge(\iota, \vec v_1) \leftarrow g_1(\iota, \vec v_1), \ldots,
g_o(\iota, \vec v_1)$;\\[0.25em]
for disjunction gate $g_\vee$ with input gates~$g_1, \ldots, g_o$, we
define for every~$1\leq k\leq o$:~~
$g_\vee(\iota, \vec v_1) \leftarrow g_k(\iota, \vec v_1).$\\[0.25em]
The predicate of the final output gate of $D$ is given
by~$g_D$.

Additionally, we construct the following rules below, thereby following $\neg \exists U. \overline{\varphi}$
over the inverse formula of~$\varphi$.
First, we guess an assignment over the variables:
\\
\myind\noindent$A(v^1, \ldots, v^{n}) \leftarrow \eneg \dot A(v^1, \ldots, v^{n}), g_D(1, v^1, \ldots, v^{n}).$
\\
\myind\noindent$\dot A(v^1, \ldots, v^{n}) \leftarrow \eneg A(v^1, \ldots, v^{n}), g_D(1, v^1, \ldots, v^{n}).$ %
\\
\myind\noindent$A(v^1, \ldots, v^{n}) \leftarrow \pneg\dot A(v^1, \ldots, v^{n}), g_D(2, v^1, \ldots, v^{n}).$
\\
\myind\noindent$\dot A(v^1, \ldots, v^{n}) \leftarrow \pneg A(v^1, \ldots, v^{n}), g_D(2, v^1, \ldots, v^{n})$. %
\\
\noindent Then, we check whether all terms are dissatisfied.\\
\noindent $\usat(1, \vec v_1, s_2, \vec v_2, s_3, \vec v_3) \leftarrow \dot A(\vec v_1), g_C(1, \vec v_1, s_2, \vec v_2, s_3, \vec v_3)$

\noindent $\usat(0, \vec v_1, s_2, \vec v_2, s_3, \vec v_3) \leftarrow  A(\vec v_1), g_C(0, \vec v_1, s_2, \vec v_2, s_3, \vec v_3)$

\noindent $\usat(s_1, \vec v_1, 1, \vec v_2, s_3, \vec v_3) \leftarrow \dot A(\vec v_2), g_C(s_1, \vec v_1, 1, \vec v_2, s_3, \vec v_3)$

\noindent $\usat(s_1, \vec v_1, 0, \vec v_2, s_3, \vec v_3) \leftarrow  A(\vec v_2), g_C(s_1, \vec v_1, 0, \vec v_2, s_3, \vec v_3)$

\noindent $\usat(s_1, \vec v_1, s_2, \vec v_2, 1, \vec v_3) \leftarrow \dot A(\vec v_3), g_C(s_1, \vec v_1, s_2, \vec v_2, 1, \vec v_3)$

\noindent $\usat(s_1, \vec v_1, s_2, \vec v_2, 0, \vec v_3) \leftarrow  A(\vec v_3), g_C(s_1, \vec v_1, s_2, \vec v_2, 0, \vec v_3)$.

We prohibit WVs with an answer set satisfying a term.

\noindent
\myind$\sat \leftarrow g_C(s_1, \vec v_1, s_2, \vec v_2, s_3, \vec v_3),
\pneg \usat(s_1, \vec v_1, s_2, \vec v_2, s_3, \vec v_3)$\\
\myind\noindent $v\leftarrow %
\eneg v, \eneg \pneg \sat.$ %

It is easy to see that there is a bijection between satisfying assignments over~$V$ of~$Q$ and world views of $\prog$.
Hardness for normal programs follows immediately from the reduction above,
since the resulting programs are already normal.

\smallskip\noindent Hardness for $\eprog \in \Disj$: We provide details in the appendix.
\end{proof}

Similarly, we conclude the following statement.
\begin{lemma}[$\star$]\label{cor:quantitative:propositional}
Let~$\eprog$ be a ground ELP and $i=2$ if $\eprog \in \Disj$, and $i=1$ if $\eprog \in \Normal \cup \Tight$, and $i=0$ if $\eprog \in \NonNeg$. Then, computing plausibility level $L(\prog, \emptyset)$ is~$\sharpd{\dpc{i}}$-complete.
\end{lemma}

If the arity is a fixed constant, we obtain the following.
\begin{lemma}[$\star$]\label{cor:quantitative:barity}
Let~$\prog$ be a non-ground ELP of bounded arity and $i=2$ if $\prog \in \Disj$, and $i=1$ if $\prog \in \Normal \cup \Tight$, and $i=0$ if $\prog \in \NonNeg$. 
Then, computing plausibility level $L(\prog, \emptyset)$ 
is~$\sharpd{\dpc{i+1}}$-complete.
\end{lemma}

\section{Non-Ground ELPs of Bounded Treewidth}\label{sec:4}

Before we discuss consequences of evaluating non-ground ELPs for
treewidth, we recall tree decompositions (TDs) for which we need the
following definition.
\begin{definition}[TD~\cite{RobertsonSeymour85}]
Let $G = (V, E)$ be a graph. A pair $\mathcal{T} = (T, \chi)$,
where $T$ is a rooted tree with 
root $\rootOf(T)$ and $\chi$ is a labeling function that maps every
node $t$ of $T$ to a subset $\chi(t) \subseteq V$ called
\emph{bag}, is a
\emph{tree decomposition
  (TD)} %
of~$G$ if (i) for each $v \in V$ some~$t$ in~$T$ exists s.t.{} $v \in \chi(t)$; (ii) for each $\{v, w\} \in E$ some $t$
in $T$ exists s.t.{ }$\{v, w\} \subseteq \chi(t)$; and (iii) for each
$r, s, t$ of $T$ s.t.\ $s$ lies on the unique path from $r$ to $t$, $\chi(r) \cap \chi(t) \subseteq \chi(s)$.    
\end{definition}
The \emph{width}
of~$\mathcal{T}$ is the %
largest bag size~minus one and the \emph{treewidth} of~$G$ is the
smallest width among all TDs of~$G$.
To simplify case distinctions in the algorithms, \emph{we
  use nice TDs in a proof}\shortversion{(see extended version),}
\longversion{, see supplemental material in Appendix~\ref{appendix:prelimns:parameterized}}
which can be computed in linear time without increasing the
width~\cite{Kloks94a}.  To capture atom (predicate) dependencies of
programs, we use the following \emph{primal graph~$G_\gprog = (V, E)$}
of a program~$\prog$ defined as follows.  For ground~$\gprog$, we let
$V \eqdef \at(\gprog)$ and $\{a,b\}\in E$ if atoms~ $a\neq b$ jointly
occur in a rule of $\gprog$, while for non-ground~$\gprog$, we let
$V \eqdef \pname(\prog)$ and $\{p_1,p_2\}\in E$ if predicates
$p_1\neq p_2$ jointly occur in a rule of $\prog$.
Tree decompositions allow us to establish tight complexity bounds for WV existence under ETH. 
To this end, we resort to quantified CSP (QCSP)\shortversion{,}
\longversion{, see Appendix~\ref{appendix:prelimns:qcsp} for definitions, }
which, intuitively, is analogous to quantified Boolean formulas, but over arbitrary finite domains instead of domain $\{0,1\}$.
To this end, we define primal graph~$P_Q$ for a QCSP~$Q$ similarly to programs, but on the formula's matrix. Further, $\expf(0,n)=n$ and $\expf(k,n)= 2^{\expf(k{-1},n)}$, $k\geq 1$, denotes the $k$-fold exponential function of $n$.  The following bounds are known.

\begin{proposition}[\citey{FichteHecherKieler21}]%
\label{prop:qcsplb}
Given any QCSP $Q$ with constraints~$C$ over finite domain $D$
and alternation depth $\ell\geq 1$, where each constraint has at most $s \geq 3$  variables.
Then, under ETH the validity of $Q$ cannot be decided in time $\expf(\ell{-}1,|D|^{o(k)})\cdot
\poly(|C|)$, where $k$ is the treewidth of %
$P_C$.
\end{proposition}

With his result at hand, we obtain the following.
\newcommand{\preds}{\at}

\begin{theorem}[$\star$]\label{thm:lower}
Let~$\prog$ be an arbitrary ELP of bounded arity~$a$ over domain size~$d=\Card{\dom(\prog)}$, where the treewidth of~$G_\prog$ is~$k$.
Furthermore, let (a)\, $i=2$ if $P \in \Disj$, (b)\, $i=1$ if $P \in \Normal \cup \Tight$, and (c)\,  $i=0$ if $P \in \NonNeg$. %
Then, under ETH, WV existence  of $\grd(\prog)$  cannot be decided in time 
$\expf(i+1,d^{{o}(k)}){\cdot}\poly(\Card{\preds(\prog)})$.
\end{theorem} %

\noindent Indeed, one can obtain a runtime adhering to this lower bound.

\begin{theorem}[$\star$]\label{thm:upper}%
Let~$\prog$ be an arbitrary ELP of bounded arity~$a$ over domain size~$d=\Card{\dom(\prog)}$, where the treewidth of~$G_\prog$ is~$k$.
Furthermore, let (a)\, $i=2$ if $P \in \Disj$ and (b)\, $i=1$ if $P \in \Normal \cup \Tight$.
Then, deciding world view existence as well as 
computing plausibility level~$L(\prog, \emptyset)$ 
of $\grd(\prog)$ can be done in 
 time $\expf(i+1,d^{\mathcal{O}(k)}){\cdot}\poly(\Card{\preds(\prog)})$.
\end{theorem}
\longversion{%
\begin{proof}[Proof (Sketch).]
    First, we compute a TD~$\mathcal{T}=(T,\chi)$ of~$G_\prog$ of width~$5\cdot k$ in time~$2^{\mathcal{O}(k)}\cdot\poly(\Card{\preds(\prog)})$, see~\cite{BodlaenderEtAl13}.
	Then, we construct a ground ELP~$G_{\grd(\prog)}$ by instantiating each bag program~$\prog_t$ by the ground program~$\grd(\prog_t)$.
	Observe that~$\Card{\grd(\prog_t)}$ is in~$d^{\mathcal{O}(a\cdot k)}=d^{\mathcal{O}(k)}$. Further, it is easy to see that~$\mathcal{T}'\eqdef (T,\chi')$ with~$\chi'(t)\eqdef \{h(\vecv D) \mid h(\vecv X)\in \preds(r), \vecv D\in\dom(\vecv X), r\in\prog_t\}$ is a TD of~$G_{\grd(\prog)}$.
	The constructed ELP can  then be solved using a DP algorithm~\cite[Listing 2]{HecherMorakWoltran20}, yielding the desired runtime result. %
\end{proof}
}

\section{Conclusion and Outlook}
We consider non-ground ELP, a popular concept to enable reasoning about answer sets.
We settle the %
complexity landscape of qualitative and quantitative reasoning tasks for non-ground ELPs, including common program fragments.
In particular, we establish that deciding whether a program admits a world view ranges between $\NEXP$ and ${\mathsmaller\NEXP^{\SigmaP3}}$. We mitigate resulting high complexity by 
bounding \longversion{the} predicate arities. Then, the complexity drops, ranging from $\SigmaP2$ to $\SigmaP4$.
In the quantitative setting, we consider levels of plausibility by quantifying the number of world views that satisfy a given query~$Q$.
We show completeness results for all common settings and classes of programs, namely, ground programs, non-ground programs, and non-ground programs of bounded arity.
We complete these results by incorporating treewidth \longversion{ into our runtime analyses}%
and establish results ranging up to four-fold exponential runtime in the treewidth, including  ETH-tight lower bounds.
\shortversion{Due to the techniques, our proofs also work for other common ELP-semantics.}
\longversion{In summary, we achieve the results for one of the most 
sophisticated semantics. In fact, our established techniques in proofs also work for other common ELP-semantics.} %

Our results contribute to several avenues for future research.  First,
we have an indication that well-known problems from the AI domains
with high complexity are amenable to ELPs.  In particular, we now have
an understanding that epistemic operators and fixed predicate arities
provide a suitable target formalism for problems on the second, third,
or fourth level of the PH, as certain variants of the diagnosis
problem \cite{EiterGottlob95b,EiterGL97}, counterfactual reasoning
\cite{EiterG96}, or default logic~\cite{FichteHecherSchindler22}.
Modeling such problems using epistemic operators might yield elegant
and instructive ASP encodings\shortversion{.}\longversion{ for such problems.}%
Second, they
indicate alternative ways for solver design: so far, standard
non-ground ELP systems ground the ELP first and then solve the
resulting ground ELP.  Our results justify that epistemic operators
can be reduced on the non-ground level without the exponential blowup.
Recall that non-ground, normal ELPs and propositional, disjunctive LPs
are of similar complexity (see Table~\ref{tab:results}).
This makes alternative grounding techniques such as lazy
grounding~\cite{WeinzierlTaupeFriedrich20} or body-decoupled
grounding~\cite{BesinHecherWoltran22} immediately accessible for ELPs.
Also, our results from Section~\ref{sec:4} build a theoretical
foundation for structure-aware ELP grounders.
This could also be interesting for structure-guided reductions to
ELP~\cite{Hecher22}.
Finally, extending the
complexity landscape of non-ground ELPs is on our agenda.  Finding
natural \NP-fragments would be interesting, since the complexity
beyond \NonNeg almost immediately jumps two levels in PH for the
Shen-Eiter semantics.
A comprehensive complexity picture in ELP similar to ASP could be of
interest in this
setting~\cite{Truszczynski11,FichteTruszczynskiWoltran15}.
We have left aside the case of \emph{maximal}
world views so far, although we expect that the complexity increases
by one level on the PH for reasoning problems.
It might also be interesting to consider complementary aspects in ELPs
where modal operators require some literals to be present in answer
sets~\cite{FichteGagglRusovac22} or where we compute quantitative aspects approximately~\cite{KabirEverardoShukla22}.
Restrictions on epistemic atoms that might be of interest or other
structural restrictions on programs, for example, fractional
hyper-treewidth~\cite{GroheMarx14}, is subject of future research.

\section*{Acknowledgements}
Authors are ordered alphabetically.
The work has been carried out while Hecher visited the Simons
Institute at UC Berkeley and was a PostDoc in the Computer Science \&
Artificial Intelligence Laboratory. at Massachusetts Institute of
Technology.
Research is supported by 
the Austrian Academy of Sciences (\"OAW), DOC Fellowship;
the Austrian Science Fund (FWF), grants P30168 and J4656;
ELLIIT funded by the Swedish government;
Humane AI Net (ICT-48-2020-RIA / 952026);
the Society for Research Funding in Lower Austria (GFF) grant
ExzF-0004; and 
Vienna Science and Technology Fund (WWTF) grants ICT19-065 and
ICT22-023.

\bibliographystyle{named}

\cleardoublepage
\appendix
\section{Additional Preliminaries}
\subsection{Propositional Logic}
We define Boolean formulas and their evaluation in the usual way and
\emph{literals} are FO-variables or their negations. For a Boolean
formula $\psi$, we denote by $\var(\psi)$ the set of variables of
formula~$\psi$. Logical operators~$\wedge$, $\vee$, $\neg$
$\rightarrow$, $\leftrightarrow$ are used in the usual meaning. A
\emph{term} is a conjunction of literals and a clause is a
\emph{disjunction} of literals.
$\psi$ is in \emph{conjunctive normal form (CNF)} if $\psi$ is a
conjunction of clauses and $\psi$ is in \emph{disjunctive normal form
  (DNF)} if $\psi$ is a disjunction of terms.

\subsection{First Order (FO) Logic}

\paragraph{Syntax}
A \emph{signature} is a finite
set~$\sigma = \SB R_1,\ldots, R_k, c_1,\ldots, c_\ell\SE$ with $k$,
$\ell \in \NAT_0$.
We call symbols~$R_1$, $\ldots$, $R_k$ \emph{relation symbols} and the
symbols~$c_1$, $\ldots$, $c_\ell$ \emph{constant symbols}.
We abbreviate the set of all %
constant symbols by~$\CONS_\sigma$.
Every $R_i$ has an arity~$\ar(R_i) \in \NAT$.
A $\sigma$-structure
$\AAA = (A, \sigma^\AAA) = (A,R^{\AAA}_1, \ldots, R^{\AAA}_k,
c^{\AAA}_1, \ldots, c^{\AAA}_\ell)$ consists of a non-empty set~$A$, the domain
of~$A$, an $\ar(R_i)$-ary relation $R_i^\AAA \subseteq A^{\ar(R_i)}$,
for every~$i \leq k$; and an element $c_j^\AAA \in A$ for every
$j \leq \ell$.
We abbreviate by $\VAR_1$ the set of all FO-variables.
We define the set~$\FO[\sigma]$ of formulas inductively as follows:
\begin{itemize}
\item[(A1)] $R(v_1, \ldots, v_k) \in \FO[\sigma]$ for all relation
  symbols~$R \in \sigma$ where
  $v_1, \ldots, v_{\ar(R)}\in \VAR_1 \cup \CONS_\sigma$.
\item[(A2)] $v = v' \in \FO[\sigma]$ for all
  $v, v' \in \VAR_1 \cup \CONS_\sigma$.
\item[(BC)] Let $\psi_1$, $\psi_2 \in \FO[\sigma]$, then
  $\neg \psi_1 \in \FO[\sigma]$,
  $\psi_1 \rightarrow \psi_2 \in \FO[\sigma]$,
  $\psi_1 \vee \psi_2 \in \FO[\sigma]$,
  $\psi_1 \leftrightarrow \psi_2 \in \FO[\sigma]$, and
  $\psi_1 \wedge \psi_2 \in \FO[\sigma]$.
\item[(Q1)] Let $\psi \in \FO[\sigma]$ and $x \in \VAR_1$,
  $\exists \psi \in \FO[\sigma]$ and $\forall \psi \FO[\sigma]$.
\end{itemize}
Formulas constructed from (A1) and (A2) are called \emph{atomic}.
If it is clear from the context, we omit stating the
signature~$\sigma$ and just assume that it is arbitrary one.

\paragraph{Interpretations}
Let $\sigma$ be a signature. A $\sigma$-interpretation
$I=(\AAA,\beta)$ consists of a $\sigma$-structure~$\AAA$ and an
assignment~$\beta: \VAR_1 \rightarrow A$, which maps every variable to
a value from~$A$. %
Let $x \in \VAR_1$, $a \in A$, and $\beta: \VAR_1 \rightarrow A$, 
then the
\emph{assignment}~$\extend{\beta}{a}{x}: \VAR_1 \rightarrow A$ is
defined as follows:
$(\extend{\beta}{a}{x})(y) \eqdef \{a \SM y = x\} \cup \{ \beta(y) \SM
y \neq x\}$.

\paragraph{Semantics}
Let $\sigma$ be a signature, $I = (\AAA, \beta)$ a
$\sigma$-interpretation and $\phi \in \FO[\sigma]$.
We define the models relation~$I \models \phi$ inductively on the
formula~$\phi$ that is constructed according to the corresponding
rules stated above:
\begin{itemize}
\item[(A1)]
  $I \models \phi$ if and only if
  $(\beta(v_1), \ldots, \beta(v_k)) \in R^\AAA$.
\item[(A2)] %
  $I \models \phi$ if and only if $\beta(v) = \beta(v')$.
\item[(BC)]
  \begin{enumerate}
  \item $\phi = \neg \phi_1$, then $I \models \phi$\\
    if and only if $I \models \phi_1$ does not hold;
  \item $\phi = (\phi_1 \vee \phi_2)$, then $I \models \phi$\\
    if and only if $I \models \phi_1$ and $I \models \phi_2$;
  \item $\phi = (\phi_1 \wedge \phi_2)$, then $I \models \phi$\\
    if and only if $I \models \phi_1$ and $I \models \phi_2$;
  \item $\phi = (\phi_1 \rightarrow \phi_2)$, then $I \models \phi$\\
    if and only if $I \models \phi_1$ then also $I \models \phi_2$;
  \item $\phi = \phi_1 \leftrightarrow \phi_2$, then $I \models \phi$\\
    if and only if $I \models \phi_1$ if and only if
    $I \models \phi_2$.
  \end{enumerate}
\item[(Q1)]
  \begin{enumerate}
  \item $\phi = \exists x \psi$, then $I \models \psi$ if and only
    if\\ there is an~$a \in A$, such that
    $(\AAA, \extend{\beta}{a}{x}) \models \psi$;
  \item $\phi = \forall x \psi$, then $I \models \psi$ if and only
    if\\
    for all~$a \in A$ it is true that
    $(\AAA, \extend{\beta}{a}{x}) \models \psi$.
  \end{enumerate}
\end{itemize}

The set~$\free(\phi)$ consists of variables that are \emph{unbounded}
in $\phi$. We define them inductively on the formula~$\phi$ that is
constructed according to the corresponding rules stated above:
\begin{itemize}
\item[(A1)] $\free(\phi) \eqdef \{v_1,\ldots, v_k \} \cap \VAR_1$,
\item[(A2)] $\free(\phi) \eqdef \{v, v'\} \cap \VAR_1$,
\item[(BC)]
  $\phi = \neg \phi_1$, then $\free(\phi) \eqdef \free(\phi_1)$ and \\
  $\phi = (\phi_1 \circ \phi_2)$, then
  $\free(\phi) \eqdef \free(\phi_1) \cup \free(\phi_2)$ where
  $\circ \in \{\vee, \wedge, \rightarrow, \leftrightarrow \}$.
\item[(Q1)] $\phi = \exists x \psi$ or $\phi = \forall x \psi$, then
  $\free(\phi) \eqdef \free(\psi) \setminus \{x\}$.
\end{itemize}
We call a formula~$\phi \in \FO[\sigma]$ \emph{sentence}, if $\phi$
has no free variables, i.e., $\free(\phi) = \emptyset$.

\paragraph{Ordered Structures}
We call a structure~$\AAA$ \emph{ordered} if its signature contains a
$2$-ary relation symbol~$<$, such that its relation~$<^\AAA$ is a
strict linear ordering on the universe~$A$,~i.e., $<^\AAA$ is
irreflexive, asymmetric, transitive, connected.
Furthermore, the relation $\first^\AAA(a)$ is true only if $a\in A$
the smallest element, and $\last^\AAA(a)$ is true only if $a \in A$ is
the largest element, and $\suc$ is the successor relation with respect
to~$<^\AAA$. In other words, $(a,b) \in \suc$ if and only if
$a <^\AAA b$ and there is no $c$ such that $a<^\AAA c <^\AAA b$.

\subsection{Second Order (SO) Logic}

\paragraph{Syntax}
The set~$\VAR_2 \eqdef \SB V^k_i \SM k,i \in \NAT \SE$ consists of all
SO-variables, sometimes also called \emph{predicate variables}.
The variable~$\VAR^k_i$ has arity $\ar(\VAR^k_i)=k$.
The set of all formulas~$\SO[\sigma]$ is defined inductively by the rules %
(A1), (A2), (BC), and (Q1) from the FO-logic above as well as %
the following rules (A3) and (Q2).
\begin{itemize}
\item[(A3)] Let $X \in \VAR_2$ of arity $k \eqdef \ar(X)$,
  $v_1,\ldots,v_k \in \VAR_1 \cup \CONS_\sigma$, then
  $X(v_1, \ldots, v_k)\in \SO[\sigma]$.
\item[(Q2)] Let $\psi \in \SO[\sigma]$ and $X \in \VAR_2$, then (i)
  $\exists X \psi \in \SO[\sigma]$ and (ii)
  $\forall X \psi \in \SO[\sigma]$.
\end{itemize}
Formulas built from (A1), (A2), and (A3) are called \emph{atomic}.

If $\phi \in \SO[\sigma]$, we mean by $\free(\phi)$ the set of all
first- or second-order variables that occur unbounded in~$\phi$.
We write $\phi(x_1, \ldots, x_s, X_1, \ldots, X_t)$ to indicate that
$\free(\phi)=\{x_1, \ldots, x_s, X_1, \ldots, X_t\}$.
We call an $\SO[\sigma]$-formula~$\phi$ a \emph{sentence} if
$\free(\phi) = \emptyset$.

\paragraph{Semantics}
The $\SO$-semantics extends semantics of $\FO$ as follows.
Let $\AAA$ be a $\sigma$-structure.
\begin{itemize}
\item[(Q2)] $\AAA\models \exists X \phi$ if and only if there is a
  relation~$X^\AAA \subseteq A^k$, such that
  $(\AAA,X^\AAA) \models \phi$.
\end{itemize}
The SO-logic is sometimes also known as~$\Sigma^1_k$.

\subsection{Quantified Constraint Satisfaction Problems (QCSPs) and Validity of Quantified Boolean Formulas ($\QBFS$)}\label{appendix:prelimns:qcsp}
We define \emph{Constraint Satisfaction Problems (CSPs)} over finite domains 
and their evaluation~\cite{FergusonOSullivan07} in the usual way. 
A CSP~$\mathcal{C}$ is a set of constraints 
and we denote by~$\var(\mathcal{C})$ %
the set of \emph{(constraint) variables} of~$\mathcal{C}$.
These constraint variables are over a (fixed) finite domain~$\mathcal{D}=\{0,\ldots\}$ 
consisting of at least two values (at least ``Boolean'').
An \emph{assignment} is a mapping~$\iota: \dom(\iota) \rightarrow \mathcal{D}$
defined for a set~$\dom(\iota)$ of variables.
We say an assignment~$\iota'$ \emph{extends}~$\iota$
(by~$\dom(\iota')\setminus \dom(\iota)$) if
$\dom(\iota') \supseteq \dom(\iota)$ and~$\iota'(y) = \iota(y)$ for
any~$y\in \dom(\iota)$.
Further, we let an \emph{assignment~$\alpha|_X$ restricted to~$X$} be
the assignment~$\alpha$ restricted to~$X$, where~$(\alpha|_X)(y)=\alpha(y)$ for~$y\in X\cap\dom(\alpha)$.
A \emph{constraint}~$C\in\mathcal{C}$ restricts certain variables~$\var(C)$ of~$C$
and contains assignments to $\mathcal{D}$.
More precisely, each~$c\in C$ with~$C\in\mathcal{C}$ is an \emph{allowed assignment}~$c: \var(C) \rightarrow \mathcal{D}$ 	
  of each variable~$v\in\var(C)$ to~$\mathcal{D}$.
Formally, $\mathcal{C}$ is \emph{satisfiable}, if there is an assignment~$A: \var(\mathcal{C}) \rightarrow \mathcal{D}$, 
called \emph{satisfying assignment}, such that for every~$C\in \mathcal{C}$
there is an allowed assignment~$c\in C$, where assignment $A|_{\var(C)}$ restricted to~$\var(C)$ equals~$c$. 

Let $\ell\geq 0$ be an integer. A \emph{quantified CSP (QCSP)}~$Q$  is of the form
$Q_{1} V_1.  Q_2 V_2.\cdots$ $Q_\ell V_\ell. \mathcal{C}$ where
$Q_i \in \{\forall, \exists\}$ for $1 \leq i \leq \ell$, $V_\ell=\exists$, and
$Q_j \neq Q_{j+1}$ for $1 \leq j \leq \ell-1$; and where $V_i$ are
disjoint, non-empty sets of constraint variables with
$\bigcup^\ell_{i=1}V_i = \var(\mathcal{C})$; and $\mathcal{C}$ is a CSP.
We call $\ell$ the \emph{quantifier rank} of~$Q$ and let $\matr(Q)\eqdef \mathcal{C}$.
The evaluation of QCSPs is defined as follows. %
Given a QCSP~$Q$ and an assignment~$\iota$, then~$Q[\iota]$ is a
QCSP that is obtained from~$Q$, where every occurrence of
any~$x\in \dom(\iota)$ in~$\matr(Q)$ is replaced by~$\iota(x)$, and
variables that do not occur in the result are removed from preceding
quantifiers accordingly.
A QCSP~$Q$ \emph{evaluates to true}, or is valid, if~$\ell=0$
and the CSP $\matr(Q)$ is satisfiable. Otherwise,
i.e., if~$\ell \neq 0$, we distinguish according to~$Q_1$.
If~$Q_1=\exists$, then~$Q$ evaluates to true if and only if there
exists an assignment~$\iota: V_1\rightarrow \mathcal{D}$ such
that~$Q[\iota]$ evaluates to true.  If~$Q_1=\forall$,
then~$Q[\iota]$ evaluates to true if for any
assignment~$\iota: V_1 \rightarrow\mathcal{D}$, $Q[\iota]$ evaluates to
true.

For brevity, we denote constraints by a formula using equality $=$
between variables and elements of the domain~$\mathcal{D}$,
negation $\neg$, which inverts (in-)equality expressions, as well as disjunction $\vee$, conjunction $\wedge$, and implication~$\rightarrow$,
which can be easily transformed into CSPs as defined above.
Further, we use for any Boolean variable~$v\in\var(Q)$ the expression ``$v$'', as a shortcut for~$v=1$.

If QCSP formula~$Q$ is over the Boolean domain, we refer to the formula simply as \emph{quantified Boolean Formula (QBF)} and $\QBFS$.

\subsection{Counting Complexity}\label{appendix:prelimns:counting}
We follow standard terminology in this area \cite{DurandHermannKolaitis05,HemaspaandraVollmer95a}.
In particular, we will make use of complexity classes preceded with the sharp-dot operator `$\#\cdot$'.
A \emph{witness} function is a function $w\colon\Sigma^*\to\mathcal P^{<\omega}(\Gamma^*)$, where $\Sigma$ and $\Gamma$ are alphabets, mapping to a finite subset of $\Gamma^*$. Such functions associate with the counting problem ``given $x\in\Sigma^*$, find $|w(x)|$''.
If $\mathcal C$ is a decision complexity class then $\#\cdot\mathcal C$ is the class of all counting problems whose  witness function $w$ satisfies (1.) $\exists$ polynomial $p$ such that for all $y\in w(x)$, we have that $|y|\leqslant p(|x|)$, and (2.) the decision problem ``given $x$ and $y$, is $y\in w(x)$?'' is in $\mathcal C$.
A \emph{parsimonious} reduction between two counting problems $\#A,\#B$ preserves the cardinality between the corresponding witness sets and is computable in polynomial time.
A \emph{subtractive} reduction between two counting problems $\#A$ and $\#B$ is composed of two functions $f,g$ between the instances of $A$ and $B$ such that $B(f(x))\subseteq B(g(x))$ and $|A(x)|=|B(g(x))|-|B(f(x))|$, where $A$ and $B$ are respective witness functions.

\newcommand{\MC}{\textsc{Mc}\xspace}
\subsection{Descriptive Complexity}

The \textsc{Model Checking} problem for a logic~$\LLL$ and
class~$\mathsf{F}$ of structures asks to decide, for a given a
sentence~$\Phi \in \LLL$ and a structure~$\AAA \in \mathsf{F}$,
whether $\AAA \models \Phi$.
\emph{Data complexity} is the complexity of the \textsc{Model
  Checking} problem assuming that $\Phi$ is fixed. More precisely, the
data complexity of \textsc{Model Checking} for logic~$\LLL$ is in the
complexity class~$\mathcal{K}$, if for every sentence~$\Phi \in \LLL$,
$L_\Phi$ can be recognized in~$\mathcal{K}$ where
$L_\Phi(\AAA) \eqdef \SB \enc(\AAA) \SM \AAA \in \mathsf{F}, \AAA
\models \Phi \SE$ and $\enc$ refers to the standard encoding,
see,~e.g.,~\cite{GradelKolaitisLibkin07}.
\emph{Combined complexity} is the complexity of
recognizing the satisfaction relation on finite structures.
More formally, the combined complexity of \textsc{Model Checking} for
logic~$\LLL$ is in complexity class~$\mathcal{K}$, if for every
sentence~$\Phi \in \LLL$, $L$ can be recognized in~$\mathcal{K}$ where
$L(\AAA,\Phi) \eqdef \SB \enc(\AAA).\enc(\Phi) \SM \AAA \in
\mathsf{F}, \AAA \models \Phi \SE$.

\subsection{Parameterized Complexity}\label{appendix:prelimns:parameterized}
\paragraph{Nice Tree Decompositions (TDs)}
For a node~$t \in N$, we say that $\type(t)$ is $\leaf$ if
$\children(t,T)=\langle \rangle$; $\join$ if
$\children(t,T) = \langle t',t''\rangle$ where
$\chi(t) = \chi(t') = \chi(t'') \neq \emptyset$; $\intr$
(``introduce'') if $\children(t,T) = \langle t'\rangle$,
$\chi(t') \subseteq \chi(t)$ and $|\chi(t)| = |\chi(t')| + 1$; $\rem$
(``remove'') if $\children(t,T) = \langle t'\rangle$,
$\chi(t') \supseteq \chi(t)$ and $|\chi(t')| = |\chi(t)| + 1$. If for
every node $t\in N$, $\type(t) \in \{ \leaf, \join, \intr, \rem\}$ and
bags of leaf nodes and the root are empty, then the TD is called
\emph{nice}.

\clearpage
\section{Additional Proofs}

\subsection{Qualitative Reasoning}\label{appendix:qualitative}
\begin{restatelemma}[lem:transform]
\begin{lemma}
  Given a sentence~%
  $\Phi\in \Sigma^1_k[\sigma]$ in CDNF and a
  finite structure~$\AAA$, %
  deciding whether $\AAA \models \Phi$  is
  (i)~$\NEXPl{k-1}$-com\-plete and
  (ii)~$\SIGMA{k+1}$-complete if every 
  predicate $R_i$ in $\Phi$ has arity at most~$m$ for some arbitrary but
  fixed constant~$m\geq 1$.
\end{lemma}
\end{restatelemma}

\begin{proof}
  Case (i)
  \emph{Membership:} Since the formulas are a special case, membership
  follows directly from Proposition~\ref{prop:socompl}.
  \emph{Hardness:} %
  Any sentence $\phi$ and structure $\AAA$ as in Proposition~\ref{prop:socompl}
  can be transformed in polynomial time into a CDNF sentence $\Phi'$ and a structure $\AAA'$
  such that $\AAA\models \Phi$ iff $\AAA'\models \Phi'$
  To this end, we can employ as in \cite{EiterGottlobGruevich96} SO-Skolemization and an ordering $<$ of the domain $A$ of $\AAA$ with 
  auxiliary relations and predicates for quantifier elimination to obtain a prenex sentence $\Phi_0 \in \Sigma^1_k[\sigma_0]$ with FO-part $\exists\vec{y}\psi_0$ resp.{} $\forall\vec{y}\psi_0$ such that
  $\AAA\models \Phi$ iff $\AAA_0\models \Phi_0$, where $\AAA_0$ extends $\AAA$.
  
  To transform $\phi_0$ into CNF for odd $k$,
  we replace sub-formulas $\psi(\vx)$ by $P_\psi(\vx)$ where we add $P_\psi$ to $R_k$ (quantify existentially), and recursively state for
  $\phi(\vx) \leftrightarrow (\phi_1(\vx) \wedge \phi_2(\vx))$ that
  $P_\phi(\vx) \leftrightarrow (P_{\phi_1} (\vx) \wedge
  P_{\phi_2}(\vx))$ etc.\ We convert the equivalences into CNF and add the conjunct $P_{\phi}(\vx)$. 
  For even $k$, we use $\exists\vec{y}\phi_0 = \neg\forall\vec{y}\neg\phi_0$ and proceed similarly ($\neg\phi_0$ is CNF).
  We establish Case~(ii) as follows.
    \emph{Membership:} %
   We proceed by induction on $k\geq 1$. In the base case $\Phi = \exists R_1\forall\vec{x}\psi$, we can in polynomial time guess extensions $R^\AAA$ for the predicates $R$ in $R_1$ and check with an coNP oracle whether $(\AAA,\beta)\models \psi$ for each assignment $\beta$ to $\vec{x}$; hence the problem is in $\Sigma^p_2$. In the induction step, we can likewise guess extensions for $R_1$, and 
   put them into the structure, i.e., extend $\AAA$ to $\AAA_1$ over $\sigma\cup R_1$; by the induction hypothesis, we can decide $\AAA_1\models \neg(\forall R_2\cdots Q_k R_k\phi)$ using an $\Sigma^p_{k}$ oracle. Thus deciding $\AAA \models \Phi$ is in $\Sigma^p_{k+1}$.
\emph{Hardness:} %
 We may assume w.l.o.g\ that a fixed sentence $\Phi\in \Sigma^1_{k+1}[\sigma]$
 as in item (ii) of Proposition~\ref{prop:socompl} is 
 w.l.o.g.{} in CDNF form; indeed, the transformation in (i) depends only on $\Phi$ but not on the structure $\AAA$. Moreover, we may assume that $\sigma$ contains
 for each $a\in A$ a single constant $c$ such that $c^\AAA = a$ denoted $c_A$ (fresh constants can be added, multiple $c$ with $c^\AAA=a$ removed).

  We reduce $\Phi$ into a formula $\Phi'$ with
  one quantifier block less, by grounding the FO-variables and
  replacing ground atoms in ``scope'' of the last block by FO-atoms; in that, we %
  represent tuples
  $(a_1,\ldots,a_r)$ of elements over $A$ 
  in the domain $A'$ of a structure $\AAA'$ over $\sigma'$ to achieve $m=1$. 
  This works in polynomial time, and $\AAA \models \Phi$ iff  $\AAA' \models \Phi'$. 

  In detail, let $\AAA=(A,\sigma^\AAA)$ be a
  $\sigma$-structure and $\Phi\in \Sigma^1_{k+1}[\sigma]$,
  $\Phi = \exists R_1 \forall R_2 \ldots Q_{k+1}R_{k+1}.\phi$, in CDNF form. Assume that  $k{+}1$ is odd, i.e., $\phi = \forall\vec{x}\psi(\vec{x})$, where
 $\psi$ is a quantifier-free CNF. 
  Let $\psi_A = \bigwedge_{\theta} \psi(\vec{x}\theta)$ be the grounding
  of $\psi$ over $\AAA$, i.e., the conjunction of all formulas
  $\psi(\vec{x}\theta)$ where $\theta$ is any replacement of each 
  variable in $\vec{x}$ with some $c_a$, $a \in A$; clearly $\psi_A$ is
  a CNF and constructible in polynomial time. We let
  $A' = A \cup \{ '\vec{a}\,' \mid
  R \in R_\Phi,  \vec{a} \in A^{ar(R)} \}$, where $R_\Phi = R_1\cup\cdots\cup R_k$, and replace each atom
  $R(\vec{c})$ in $\psi_A$, $R \in R_\Phi$, by
  $R'('\vec{c}\,')$, where $R'$ is a fresh unary predicate variable for
  $R$, and $'\vec{c}\,'$ is a fresh constant symbol, resulting in $\psi'_A$. Furthermore, let $\Phi_A' = \exists R'_1 \forall R'_2 \ldots Q_k R'_k Q_{k+1}R'_{k+1}.\psi'_{C_A}$.
    We then have $\AAA \models \Phi$  iff $\AAA \models \exists R_1 \forall R_2 \ldots Q_{k+1}R_{k+1}.\psi_A$  iff
    $\BBB \models \Phi_A'$,  where $\mathcal{B} =(B,\sigma_{\psi_A}^\mathcal{B})$ is the $\sigma_{\psi_A}$-structure, $\sigma_{\psi_A} =$ $ \sigma\cup \{ '\vec{c}\,' \mid R(\vec{c})\in \psi_A\}$ such that $B= A$ and $\sigma_{\psi_A}^\BBB$ extends $\sigma^\AAA$ with mapping each $'\vec{c}\,'$ to $'(c_1^\AAA,\ldots,c_n^\AAA)'$ where $\vec{c}= c_1,\ldots,c_n)$. 

   It remains to eliminate $Q_{k+1}R_{k+1}$, i.e.,
  $\exists R'_{k+1}$ from $\Phi_A'$. To this end, we use a fresh
  unary relation $\gat$ that holds a single, arbitrary element $a_{\gat} \in
  A'$; intuitively, $\gat(x)$ evaluates to true if $x$ is assigned $a$ 
  and to false otherwise. We modify $\Phi_A'$ as follows. For each atom $R'('\vec{c}\,')$ in
  $\Phi_A'$, where $R' \in R'_{k+1}$,  we introduce a variable
  $x_{R'('\vec{c}\,')}$ and (a) replace $R'('\vec{c}\,')$ by
  $\gat(x_{R'('\vec{c}\,')})$ and (b) add $\exists x_{R'('\vec{c}\,')}$
  after $\exists R_{k+1}$. Finally, we drop $\exists R_{k+1}$.
  The resulting formula $\Phi'=\exists R'_1 \forall R'_2 \ldots
  Q_k R'_k.\phi'$ is in CDNF, and for the $\sigma'$-structure
  $\AAA' = (A',{\sigma'}^{\AAA'})$ where $\sigma' $ extends $\sigma_{\psi_A}$ with mapping $\gat$ to $\{ a_{\gat}\}$, we then have $\BBB \models \Phi_A'$ iff  
  $\AAA' \models \Phi'$. 

  Since $\Phi'$ is of the desired form and 
  all steps are feasible in polynomial time, $\Sigma^p_{k+1}$-hardness
  for odd $k+1$ and $m=1$ is shown. For even $k+1$, we can proceed similarly as
  in (i), using  $\exists\vec{x}\psi = \neg\forall\vec{x}\neg\psi$.
  This shows $\Sigma^p_{k+!}$- hardness, and thus Case~(ii) of
  the lemma.
 
 Furthermore, we can eliminate FO-atoms over $\sigma$ from $\psi_A$ by
 evaluating them over $\AAA$; then we can set $\sigma' = \{ \gat\}$,
 and all predicate and relation symbols in $\Phi'$ have arity~1.
\end{proof}

\begin{restatelemma}[lem:epihorn]
\begin{lemma}
  Let~$\prog \in \NonNeg$. %
  Then, deciding whether~$\prog$ admits a world view is
  in \Ptime{} if $P$ is ground,
  \coNP-complete if $P$ is non-ground and has bounded arity, and
  \EXP-complete if $P$ is non-ground.
\end{lemma}
\end{restatelemma}
\begin{proof}[Proof (Sketch).]
\emph{Membership:} %
  Let $P \in \NonNeg$. We replace each occurrence of $\kop p$ and $\mop p$ in $P$ by $p$.
  Then, we simply check consistency of $P$.
  Since $P$ is Horn, the program has a either a unique answer set or none.
  In the latter case, we are trivially done as the program has no world view.
  In the former case, the program has a unique least model (see e.g.,~\cite{Truszczynski11b}).
  Thus, we set $\kop p$ and $\mop p$ according to the value 
  in the unique least model.
  If $P$ is non-ground and has bounded arity, checking consistency of $P$ is in \coNP~\cite{EiterFaberFink07}.
  Finally, we use the same argument as in the proof of Theorem~\ref{thm:quali} 
  for the non-ground case. Grounding an ELP~$\eprog$ leads to an exponentially 
  larger program $\grd(\eprog)$ and $\Ptime$ becomes $\EXP$, cf.~\cite{GottlobLeoneVeith99}.\\
  \emph{Hardness:} %
     For both non-ground cases, hardness follows immediately from the respective complexity results in ASP~\cite{EiterFaberFink07,GottlobLeoneVeith99}.
\end{proof}

\subsection{Quantitative Reasoning}\label{appendix:quantiative}

\begin{restatetheorem}[thm:quantitative]
\begin{theorem}
  Let~$\eprog$ be an ELP and 
  $i=2$ if $\eprog \in \Disj$, 
  $i=1$ if $\eprog \in \Normal \cup \Tight$, 
  and $i=0$ if $\eprog \in \NonNeg$.
  Then, computing plausibility level~$L(\eprog,\emptyset)$ is
  $\sharped{\SIGMA{i}}$-complete. 
\end{theorem}    
\end{restatetheorem}

\begin{proof}
We provide membership for all fragments and the hardness for fragments $\NonNeg$, $\Normal$, and $\Tight$ already in the main part of the paper.
It remains to show the case $\eprog \in \Disj$.
We reduce from~$\#\QBFSAT_2$, taking a QBF $Q=\forall U. \exists W. \varphi$ over free variables~$V$,
defined by a clause circuit~$C$ and a quantifier circuit~$D$ over~$n$ many input gates,
and constructing an ELP~$\prog$.
First, $C$ and~$D$ is inductively constructed, as in the hardness cases for the other cases in the main part.

Additionally, we construct the following rules below, thereby following $\neg \exists U. \forall W. \overline{\varphi}$
over the inverse formula of~$\varphi$.
First, we guess an assignment over the variables:

$A(v^1, \ldots, v^{n}) \leftarrow \eneg\dot A(v^1, \ldots, v^{n}), g_D(1, v^1, \ldots, v^{n}).$

$\dot A(v^1, \ldots, v^{n}) \leftarrow \eneg A(v^1, \ldots, v^{n}), g_D(1, v^1, \ldots, v^{n}).$ %

$A(v^1, \ldots, v^{n}) \leftarrow \pneg\dot A(v^1, \ldots, v^{n}), g_D(2, v^1, \ldots, v^{n}).$

$\dot A(v^1, \ldots, v^{n}) \leftarrow \pneg A(v^1, \ldots, v^{n}), g_D(2, v^1, \ldots, v^{n}).$ %

$A(v^1, \ldots, v^{n}) \vee \dot A(v^1, \ldots, v^{n}) \leftarrow g_D(3, v^1, \ldots, v^{n}).$ %

To check whether there is a dissatisfied clause, one could use~$2^3=8$ rules corresponding to~$8$ potential combinations of signs. However, we use the following shortcuts:

\noindent $L(0, \vec v) \leftarrow \dot A(\vec v). $  $L(1, \vec v) \leftarrow A(\vec v). $

\noindent $\usat %
\hspace{-.3em}\leftarrow\hspace{-.3em} g_C(s_1, \vec v_1, s_2, \vec v_2, s_3, \vec v_3),L(s_1, L\vec v_1), L(s_2, \vec v_2), L(s_3,\vec
v_3).$

The universally quantified variables  ($W$) are saturated in case a clause is dissatisfied.

\noindent $A(v^1, \ldots, v^{n}) \leftarrow g_D(3, v^1, \ldots, v^{n}), \usat.$ %

\noindent $\dot A(v^1, \ldots, v^{n}) \leftarrow  g_D(3, v^1, \ldots, v^{n}), \usat.$ %

Then, it is ensured that there is no answer set for the current WV candidate that does not dissatisfy any clause.

$v\leftarrow \eneg v, \eneg \neg \usat.$

It is easy to see that indeed there is a bijection between satisfying assignments over~$V$ of~$Q$ and world views of~$\eprog$. %
This concludes the hardness proof for the case~$\eprog \in \Disj$.
\end{proof}

\begin{restatetheorem}[cor:quantitative:propositional]
\begin{lemma}
Let~$\eprog$ be a ground ELP and $i=2$ if $\eprog \in \Disj$, and $i=1$ if $\eprog \in \Normal \cup \Tight$, and $i=0$ if $\eprog \in \NonNeg$. Then, computing plausibility level $L(\prog, \emptyset)$ is~$\sharpd{\dpc{i}}$-complete.
\end{lemma}
\end{restatetheorem}

\begin{proof}[Proof (Sketch).]
The result for~$i{=}2$ follows directly from~\cite[Theorem~4]{ShenEiter16}, which shows the complexity of checking a witness (i.e., a WV). 
For~$i=1$, we briefly sketch a hardness reduction from~$\#(\QBFS_1 \wedge \co\hy\QBFS_1)$, as membership is trivial.
To this end, we take an arbitrary formula of the form~$\exists Y. \varphi(X,Y) \wedge \neg(\exists Z. \psi(X,Z))$
with~$\varphi$ and~$\psi$ in CNF.
We construct an epistemic program
$P=\{x \leftarrow \eneg \hat x, \hat x \leftarrow \eneg x \mid x\in X\} \cup \{a \leftarrow \neg \hat a, \hat a \leftarrow \neg a \mid a\in Y\cup Z\} \cup \{q\leftarrow \eneg q, \eneg  \neg sat_\psi, \leftarrow \neg sat_\varphi\}$.
This program models epistemic guesses over variables in~$X$ and regular guesses over~$Y\cup Z$. Further, it ensures the existential and universal quantification. It remains to define~$sat_\psi$ and $sat_\varphi$.
Let~$\varphi=\{c_1,\ldots,c_{|\varphi|}\}$, $\psi=\{c_1,\ldots,c_{|\psi|}\}$ and for any literal~$\ell$, let~$f(\ell)=\ell$ if~$\ell$ is a variable and~$f(\ell)=\hat v$ with~$v=\var(\ell)$ otherwise.
Then, we define~$P_\varphi=\{sat_\varphi \leftarrow sat_{c_1}, \ldots sat_{c_{|\varphi|}}\} \cup \{sat_{c_i} \leftarrow f(\ell) \mid c_i \in \varphi, \ell\in c\}$.
Analogously, we define~$P_\psi=\{sat_\psi \leftarrow sat_{d_1}, \ldots sat_{c_{|\psi|}}\} \cup \{sat_{d_i} \leftarrow f(\ell) \mid d_i \in \psi, \ell\in c\}$.
It is easy to see that there is a bijection between WVs of~$P\cup P_\varphi \cup P_\psi$ and satisfying assignments over~$X$ of~$\exists Y. \varphi(X,Y) \wedge \neg(\exists Z. \psi(X,Z))$. 

The case $i{=}0$ works similarly, where the hardness uses ideas from Theorem~\ref{thm:quantitative} (restricted to ground programs), but we reduce from~$\#(\QBFS_0 \wedge \co\hy\QBFS_0)$. For modeling $\co\hy\QBFS_0$, the idea is to define satisfiability of a clause as well as satisfiability of the whole formula (conjunction over every clause). Then, we  prohibit satisfiability of the formula using a constraint without the need of default negation.
\end{proof}

\begin{restatecorollary}[cor:quantitative:barity]
\begin{lemma}
Let~$\prog$ be a non-ground ELP of constant arity and $i=2$ if $\prog \in \Disj$, and $i=1$ if $\prog \in \Normal \cup \Tight$, and $i=0$ if $\prog \in \NonNeg$. Then, computing plausibility level~$L(\eprog,\emptyset)$  is~$\sharpd{\dpc{i+1}}$-complete.
\end{lemma}    
\end{restatecorollary}

\begin{proof}
\emph{Membership:}
    We can employ the membership construction for the bounded arity case as done in the proof of Theorem~\ref{thm:quali}. 
    However, we ask for finding the number of WVs. Thus, by the definition of the dot-class~$\sharpd{C}$, the class~$C$ states the complexity %
    of verifying a witness (i.e., a WV).
    Since verifying a WV requires a call to~$\Sigma_{i+1}^\Ptime$ and a call to $\Pi_{i+1}^\Ptime$, this results in $\dpc{i+1}$~\cite{ShenEiter16}. Hence%
    , we obtain membership for $\sharpd{\dpc{i+1}}$.

\emph{Hardness:}
    The constructions work similarly to the proof of Lemma~\ref{cor:quantitative:propositional} by observing that with non-ground programs of bounded arity, we can evaluate a conjunctive query, see, e.g.,~\cite{EiterFaberFink07}.
    We briefly sketch hardness for
    the case $i{=}2$. To this end, we reduce from $\exists Y. \forall Z. \exists Q. \varphi(X,Y,Z,Q)) \wedge \neg(\exists U. \forall V. \exists W. \psi(X,U,V,W))$ with~$\varphi$ and~$\psi$ in 3-CNF.
We construct an epistemic program
$P=\{x \leftarrow \eneg \hat x, \hat x \leftarrow \eneg x \mid x\in X\} \cup \{a \vee \hat a \mid a\in Y\cup Z \cup Q \cup U \cup V \cup W\} \cup \{q\leftarrow \eneg q, \eneg \neg sat_\psi, \leftarrow \neg sat_\varphi\}$.
This program models epistemic guesses over variables in~$X$ and guesses over~$Y\cup Z\cup Q \cup U \cup V \cup W$. 
We saturate over~$Z$ and~$V$ by
$P_{sat}=\{a \leftarrow sat_\varphi, \hat a \leftarrow sat_\varphi \mid a\in Z\} \cup \{a \leftarrow sat_\psi, \hat a \leftarrow sat_\psi \mid a\in V\}$.

It remains to define~$sat_\psi$ and $sat_\varphi$.
As above, let~$\varphi=\{c_1,\ldots,c_{|\varphi|}\}$ and for any~$c_i\in\varphi$ and~$\ell\in c_i$, let~$f(\ell)=\ell$ if~$\ell$ is a variable and~$f(\ell)=\hat v$ with~$v=\var(\ell)$ otherwise.

Further, let~$\prec$ be any arbitrary, but fixed total order among all variables in~$\varphi\cup\psi$.
We refer to the $\prec$-ordered vector of variables in a set~$Q$ that are contained in a clause~$c_i\in\varphi$, by~$\vec Q_{i}$.
Observe that since both formulas~$\varphi$ and~$\psi$ are in 3-CNF, for any of these vectors $|\vec Z_i|\leq 3$.
For a literal~$\ell$, $pol(\ell)$ refers to its polarity, which is~$1$ if~$\ell$ is a variable and~$0$ otherwise.
We define 
$P_\varphi{=}$\\
\mbox{~}$\quad\{sat_\varphi {\leftarrow} sat_{c_1}(\vec Q_1), \ldots sat_{c_{|\varphi|}}(\vec Q_{|\varphi|})\} \,{\cup}$\\
\mbox{~}$\quad\{bin(0), bin(1)\}\,\cup $\\
\mbox{~}$\quad\{sat_{c_i}(\vec Q_i) \leftarrow f(\ell), bin(q_1), \ldots, bin(q_o) \;\mid $\\ 
\mbox{~}$\qquad\qquad\qquad\qquad c_i \in \varphi, \ell\in c, \var(\ell)\in X\cup Y \cup Z, $\\
\mbox{~}$\qquad\qquad\qquad\qquad\qquad\qquad\qquad\quad\; \vec Q_i=\langle q_1, \ldots, q_o \rangle\} \cup$\\ 
\mbox{~}$\quad \{sat_{c_i}(\langle q_1, \ldots, q_{j-1}, pol(q_j), q_{j+1}, \ldots, q_o \rangle) \leftarrow $\\
\mbox{~}$\qquad\qquad\qquad\qquad\qquad\qquad\qquad bin(q_1), \ldots, bin(q_o) \mid$\\
\mbox{~}$\qquad\qquad\qquad\qquad\qquad\qquad\quad c_i \in \varphi, \ell\in c, \var(\ell)=q_j, $\\
\mbox{~}$\qquad\qquad\qquad\qquad\qquad\qquad\qquad\qquad \vec Q_i=\langle q_1, \ldots, q_o \rangle\}$.

\noindent Analogously to above, we define~$P_{\psi}$ (see also the proof of Lemma~\ref{cor:quantitative:propositional}.
\noindent Indeed, there is a bijection between WVs of $P\cup P_{sat} \cup P_\varphi \cup P_\psi$ and satisfying assignments over~$X$ of~$\exists Y. \forall Z. \exists Q. \varphi(X,Y,Z,Q) \wedge \neg(\exists U. \forall V. \exists W. \psi(X,U,V,W))$. 
\end{proof}

\subsection{Parameterized Complexity}
\subsubsection{Lower Bound}
\begin{restatetheorem}[thm:lower]
\begin{theorem}
Let~$\prog$ be an arbitrary ELP of bounded arity~$a$ over domain size~$d=\Card{\dom(\prog)}$, where the treewidth of~$G_\prog$ is~$k$.
Furthermore, let $i=2$ if $P \in \Disj$, $i=1$ if $P \in \Normal \cup \Tight$, and $i=0$ if $P \in \NonNeg$.\\ 
Then, under ETH, WV existence  of $\grd(\prog)$  cannot be decided in time 
$\expf(i+1,d^{{o}(k)}){\cdot}\poly(\Card{\preds(\prog)})$.
\end{theorem} %
\end{restatetheorem}%

\begin{proof}[Proof (Sketch).]
    We establish the statement for the hardest case ($\eprog \in \Disj$). The remaining cases work analogously.
    We show this lower bound by contradiction, where we reduce from the problem $\QCSPVAL{}$ of quantifier rank~$4$. 
    For any QCSP instance~$Q=\exists V_1. \forall V_2. \exists V_3. \forall V_4. \mathcal{C}$ over 
    domain size~$\dom(Q)$ it is known~\cite{FichteHecherKieler21} that~$Q$ cannot be decided in 
    time $\expf(3,{d'}^{{o}(k')}){\cdot}\poly(\Card{\var(Q)})$ where~$k'$
    is the treewidth of~$G_Q$ and~$d'=\Card{\dom(Q)}$.
	Without loss of generality, see Proposition~\ref{prop:qcsplb},
	we assume that each constraint~$C_i\in\mathcal{C}$ is over three variables, i.e., arity~$a'=3$. 
	Further, let~$\mathcal{T}'=(T,\chi')$ be a nice  TD 
        of~$G_Q$
	and for every node~$t$ of~$T$, let~$C_t$ be the %
	set of \emph{bag constraints} using only variables in~$
 \chi(t)$. %
	Without loss of generality, we assume that~$\Card{C_t}\leq 1$,
	which can be easily achieved by adding intermediate auxiliary nodes (copies of nodes) to~$\mathcal{T}'$.
 Further, the \emph{bag assignments~$\varphi_t$} are a set 
 of allowed assignments in~$C_t$. Similarly to~$\varphi_t$, we may assume without loss of generality that~$\Card{\varphi_t}\leq 1$.

 We will simulate the formula~$\exists V_1. \neg(\exists V_2. \forall V_3. \exists V_4.) \neg\mathcal{C}$, thereby defining a program $\prog$ as the union of programs below.
 
 First, we guess assignments for $V_1$ via epistemic negation
 and assignments over~$V_2\cup V_3$ using disjunction. Note that disjunction is required for those over~$V_3$, subject to saturation.
 \vspace{-1.5em}
	\begin{flalign*}
	P_g\,{=}\,\SB &\hat v \gets \eneg\ v,\  %
    v \gets  \eneg\ \hat v \SM v\in V_1\SE.\\
	P'_{g}\,{=}\,\SB & \textstyle\bigvee_{d\in \dom(v)} v(d) \gets~ \mid v\in V_2\cup V_3\SE,\\
 \end{flalign*}~\\[-3em]
Then, we define the unsatisfiability of an assignment.
 \vspace{-.35em}
	\begin{flalign*}
	&P_{\varphi_t}{=}\SB \usat_{\varphi_t}(\langle\theta\rangle)\gets x_i(v_i), \usat_{\varphi_{t'}}(\langle\theta\rangle)  \mid t\text{ in }T, t'\in\\  %
	&\;\children(t), C_t{=}C_{t'}{\neq\emptyset}, \theta: V_4(C_t) \rightarrow \Delta, %
	\var(\varphi_t)\setminus V_4\\	
	&\;{\,=\,}\{x_1, \ldots, x_o\},\{v_1,\ldots, v_o\} \subseteq \Delta, %
	1\leq i\leq o, v_i\neq \varphi_t(x_i)\SE.%
	\end{flalign*} 
 This is followed by defining  unsatisfiability of a constraint.
 \vspace{-.35em}
	\begin{flalign*}
	P_{C_t}\,{=}\,\SB \usat_{C_{t}}(\langle\theta\rangle) \gets \usat_{\varphi_{t}}(\langle \theta\rangle) \mid %
     t \in T\,{\cap}\,\children(t^*), %
	 \\
	\; C_{t^*}{\neq} C_t\},\\
  P'_{C_t}\,{=}\,\SB\usat_{C_t}(\langle\theta\rangle)\gets \mid t\text{ in }T, \theta: V_4(C_t) \to \Delta,  %
	\\
	\theta\notin \{\varphi|_{V_4(C_t)} \mid \varphi\in C_t\}, C_t\neq \emptyset\SE,
 \end{flalign*}
 Finally, we define unsatisfiability, consistency, and saturate.
 \vspace{-.35em}
 \begin{flalign*}
 P_{usat}\,{=}\,\SB\usat{}\gets \usat_{C_t}(\langle V_4(C_t) \rangle) \mid t\text{ in }T\}, \\
	P_{cons}\,{=}\,\SB v \hspace{-.25em}\gets\hspace{-.25em} \eneg v,\eneg \neg \usat\}, \\ %
	P_{saturate}\,{=}\,\SB v(d) \hspace{-.3em}\gets\hspace{-.3em}\usat \SM v\in V_3, d\in\dom(v)\SE.
	\end{flalign*}%

 Observe that~$\prog$ admits a world view %
 if and only if~$Q$ is evaluates to true. 
	Further, observe that~$d$ is linear in~$d'$.

	Then, we construct a TD~$\mathcal{T}\eqdef(T,\chi)$
	of~$G_\prog$, whose width~$k$ is linear in~$k'$ and where~$\chi$ is defined for a node~$t$ of~$T$ by
	$\chi(t)\eqdef \chi'(t) \cup \{\usat, \usat_{\varphi_t}, \usat_{\varphi_{t'}}, \usat_{C_t} \mid t'\in\children(t)\}$.
	By construction, the width of~$\mathcal{T}$ is in~$\mathcal{O}(k')$,
	since~$\Card{\children(t)}\leq 2$, $\varphi_t\leq 1$, and $\Card{C_t}\leq 1$.

	Now assume towards a contradiction that consistency of~$\prog$ can be decided in %
    time $\expf(2,d^{o(k)}){\cdot}\poly(\Card{\preds(\prog)})$.
 Then, however,
	via applying the reduction above,
	we can also decide~$Q$ in 
    time~$\expf(2,{d'}^{o(k')}){\cdot}\poly(\Card{\var(Q)})$,
    which by Proposition~\ref{prop:qcsplb} contradicts ETH. %
    In consequence, we established the result for $\eprog \in \Disj$.
\end{proof}

\end{document}